\documentclass[10pt,leqno]{amsart}
\usepackage{graphicx}
\makeatother

\usepackage[indentafter]{titlesec}
\titleformat{name=\section}{}{\thetitle.}{0.8em}{\centering\scshape }
\titleformat{name=\subsection}{}{\thetitle.}{0.5em}{ \scshape}
\titleformat{name=\subsubsection}[runin]{}{\thetitle.}{0.5em}{\itshape}[.]
\titleformat{name=\paragraph,numberless}[runin]{}{}{0em}{}[.]
\titlespacing{\paragraph}{0em}{0em}{0.5em}
\titleformat{name=\subparagraph,numberless}[runin]{}{}{0em}{}[.]
\titlespacing{\subparagraph}{0em}{0em}{0.5em}

\usepackage{indentfirst,csquotes}
\usepackage{babel}
\usepackage{blindtext}
\usepackage{enumerate}
\usepackage{hyperref}
\usepackage{parskip}

\usepackage{xpatch}

\usepackage{amssymb,amsthm,amsmath}
\usepackage{xcolor,paralist,hyperref,titlesec,fancyhdr,etoolbox}
\newtheorem{theorem}{Theorem}[]
\newtheorem{remark}{Remark}
\newtheorem{definition}[theorem]{Definition}

\newtheorem{lemma}[theorem]{Lemma}
\newtheorem{proposition}[theorem]{Proposition}
\newtheorem{corollary}[theorem]{Corollary}

\newcommand{\genlegendre}[4]{%
  \genfrac{(}{)}{}{#1}{#3}{#4}%
  \if\relax\detokenize{#2}\relax\else_{\!#2}\fi
}
\newcommand{\legendre}[3][]{\genlegendre{}{#1}{#2}{#3}}

\newcommand{\N}{\mathbb{N}}

\hypersetup{ colorlinks=true, linkcolor=black, filecolor=black, urlcolor=black }

\usepackage{lipsum}
\usepackage{float}

\begin{document}
\begin{sloppypar}
\title{Average Case Error Estimates of the Strong Lucas Test} 

\author{Semira Einsele \and Kenneth Paterson}
\address{Semira Einsele, Department of Mathematics and Computer Science, FU Berlin, Germany  \vspace{-0.3cm}}
\email{semira.einsele@fu-berlin.de \vspace{-0.35cm}}
\address{Kenneth Paterson, Department of  Computer Science, ETH Zürich, Switzerland \vspace{-0.3cm}}
\email{kenny.paterson@inf.ethz.ch}
\maketitle
\kern-3em

\begin{abstract}
Reliable probabilistic primality tests are fundamental in public-key cryptography. In adversarial scenarios a composite with a high probability of passing a specific primality test could be chosen. In such cases we need worst-case error estimates of the test. However, in many scenarios the numbers are randomly chosen and thus have significantly smaller error probability. We are hence interested in average case error estimates.
In this paper, we establish such bounds for the strong Lucas primality test, as there exist only worst-case, but no average case error bounds. This allows us to use this test with more confidence. Let us examine an algorithm that draws odd $k$-bit integers uniformly and independently, runs $t$ independent iterations of the strong Lucas test with randomly chosen parameters and outputs the first number that passes all $t$ consecutive rounds. We attain numerical upper bounds on the probability that a composite is returned. Moreover, we examine a slight modification of this algorithm that only considers integers that are not divisible by small primes, yielding improved bounds. In addition, we classify the numbers that contribute most to our estimate.
\end{abstract} 

\section{Introduction}\label{sec1}
Prime generation is a basic cryptographic operation as most modern public-key cryptosystems make use of large prime numbers, either as secret or public parameters. To generate large primes, one common approach is to randomly choose integers of appropriate size and then test them for primality until a prime is found. This encourages us to find primality testing algorithms that are polynomial in complexity. While there are several sophisticated deterministic general-purpose algorithms available for primality testing, their efficiency is often not sufficient for practical applications. Therefore, \emph{probabilistic primality} tests are commonly used in practice. These tests are randomized primality testing algorithms that have a small probability of classifying a composite number as prime. Nearly all known probabilistic primality tests are based on a similar principle. From the input number $n$, one defines an Abelian group and then checks if the group structure we expect to see if $n$ were prime is present. If $n$ is composite, this structure is not always, but often absent. Running multiple independent rounds of the test can enhance its strength. In this paper, we will refer to both probabilistic and deterministic tests as primality tests. 

There are certain scenarios where the public-key parameters, such as in the Diffie-Hellman key exchange protocol, may have been chosen by an adversary. The integer could be constructed in such a way that it has a high probability of falsely being declared as prime by a specific primality test, even though it is actually composite. Therefore, in such cases, it is crucial for the primality test to have a low \textit{worst case} error probability  (\cite{galbraith2019safety}, \cite{albrecht2018prime}).
However, for many other applications where the integer is randomly chosen, such as prime generation, it is more important to know how the test behaves in the \textit{average case}, as it seems that most randomly chosen composites would be accepted with probability much smaller than the so called worst-case numbers. More formally, let us examine an algorithm that repeatedly chooses random odd $k$-bit integers and runs $t$ iterations of the primality test on each candidate. If the candidate passes all $t$ consecutive iterations, the algorithm returns that number, otherwise another randomly chosen odd $k$-bit integer is selected and tested. The algorithm ends when a number that passes all $t$ consecutive rounds is found. The error probability that this algorithm returns a composite is called the average case error probability.

There are many probabilistic primality tests, and among them is a class of primality tests that is based on the so-called Lucas sequences. Let $D$, $P$ and $Q$ be integers such that $D=P^2-4Q$ is non-zero and $P>0.$ Let $U_0(P,Q)=0$, $U_1(P,Q)=1$, $V_0(P,Q)=2$ and $V_1(P,Q)=P.$
The \textit{Lucas sequences} $U_n(P, Q)$ and $V_n(P,Q)$ associated with the parameters $P$, $Q$ are defined recursively for $n \geq 2$ by
\begin{align*}
    U_n(P,Q)&=PU_{n-1}(P,Q)-QU_{n-2}(P,Q),\\
    V_n(P,Q)&=PV_{n-1}(P,Q)-QV_{n-2}(P,Q).
\end{align*}

For an integer $n$, let $\epsilon(n)$ denote the Jacobi symbol $\legendre{D}{n}$. In 1980 Baillie and Wagstaff \cite{Lucas-Baillie} gave a thorough treatment of the use of Lucas sequences in primality testing and examined various congruences that hold for prime numbers. 

\begin{theorem}\label{defslpsp}
Let $P$ and $Q$ be integers and $D= P^2-4Q$. Let $p$ be a prime number not dividing $2QD$. Write $p-\epsilon(p)=2^\kappa q$, where $q$ is odd. Then
\begin{equation}\label{thrm-slpsp} \text{ either } p \mid U_q 
    \text{ or } 
     p \mid V_{2^iq} \text{ for some $0\leq i <\kappa$. }
\end{equation}
\end{theorem}

From this theorem we can derive a primality test for an integer $n$ with a fixed $D$ by checking property (\ref{thrm-slpsp}) for several uniformly at random chosen bases $(P,Q)$, where $1\leq P,Q \leq n$, $\gcd(Q,n)=1$ and $P=D^2-4Q$. This test is called \textit{the strong Lucas test.} If (\ref{thrm-slpsp}) does not hold for some base $(P,Q)$, then $n$ is certainly composite. We call such $(P,Q)$ a \textit{witness for compositeness} using the strong Lucas test, which is a short proof that $n$ is composite. However, if (\ref{thrm-slpsp}) is true for several bases, even though this does not serve as a proof, it is very likely that $n$ is a prime.

Arnault \cite{Rabin-Mon-Lucas} demonstrated that the worst-case numbers of the strong Lucas test occur for twin-prime products, which are products of two primes with a prime-gap of 2. In such cases, half of the bases $(P,Q)$ used in the test declare the integer to be prime. For integers that are not certain twin-prime products, only $4/15$ of the bases pass the test. These results serve as the worst-case error estimate of the strong Lucas test. Luckily, excluding twin-prime products does not impose significant restrictions, as these numbers can be easily detected by running Newton's method for square roots prior to conducting the actual test. 

From Arnault's result it may be tempting to directly conclude that for non-twin-prime products the average case estimate for $t$ rounds of the strong Lucas test is $(4/15)^t$.
This reasoning is wrong, as the following discussion shows. For any $k$, denote $M_k$ the set of odd $k$-bit integers. Let $t\geq 1$ be fixed and choose $k$ sufficiently large such that the density of the primes in $M_k$ is much less than $ (4/15)^t$. Assume that for most composites in $M_k$ the probability that the integer that we test for primality passes a test with randomly chosen bases is about $ 4/15$. Then, of course, the probability of it passing $t$ independent tests is about $ (4/15)^t$. Suppose that we have an integer $n$ from $M_k$ that that passes $t$ tests. Since we are assuming that primes in $M_k$ are scarce, it will be much more likely that $n$ is composite rather than prime, so the average case estimate would be close to $1$. Naturally the average case estimate is much smaller than the worst-case, so we need a different argument for obtaining worst-case bounds.

The Miller-Rabin test is a widely used probabilistic primality test, as it is well-studied and easy to implement. Rabin \cite{Rabin} and Monier \cite{Monier} individually and almost simultaneously established  worst-case error bounds for the Miller-Rabin test, Damgård, Landrock and Pomerance \cite{DamEtAl} established average case error bounds, which justifies the trust in this test by the cryptographic community.

For the strong Lucas test on the contrary only worst-case upper bounds are known and randomly choosing worst-case numbers is rather an unusual occurrence. We are hence concerned with finding average error estimates for the strong Lucas test. Such results would allow us to employ this test with more confidence in practice. For this, we consider an algorithm that draws odd $k$-bit integers independently from the uniform distribution, runs $t$ independent iterations of the strong Lucas test with randomly chosen parameters on each candidate, and outputs the first one that passes all $t$ consecutive rounds. Let $q_{k,t}$ denote the probability that a number outputted by this algorithm is composite.

In this paper, we conduct a thorough analysis of this error probability, and derive explicit numerical upper bounds for $q_{k,t}$. These bounds are obtained by adapting the methods used in \cite{DamEtAl} for the strong Lucas case. We also observe that by incorporating trial division by small primes before running the strong Lucas test, we achieve notable improvements in the error estimates. The inclusion of trial division is not a restrictive assumption but rather a common practice in cryptographic software to enhance the runtime efficiency of the tests. Therefore, this assumption is often naturally implemented without incurring additional computational costs. Let $q_{k,l,t}$ denote the probability of returning a composite of the modified algorithm, which includes the condition of considering only integers that are not divisible by the first odd $l$ primes. Furthermore, we identify the numbers that contribute most to our probability estimate in the strong Lucas test and realize that amongst others, special types of Lucas-Carmichael numbers belong to this set. The main results of the paper are
\begin{align*}
    q_{k,1} <& \; \log(k)k^24^{2.3-\sqrt{k}} \text{  \; for } k \geq 2,\\
    q_{k,l,1} <& \; k^24^{1.8-\sqrt{k}}  \rho_l^{2\sqrt{k-1}-2} \text{ \; for } k\geq 2, l\in \N,\\
    q_{k,t} <& \;\log(k)^tk^{3/2}\frac{2^t}{\sqrt{t}}4^{2.12-\sqrt{tk}} \text{  \; for } k \geq 79,\; 3 \leq t \leq k/9 \text{ or } k\geq 88, \; t = 2,\\
    q_{k,l,t} <& \;  4^{1.72-\sqrt{tk}}k^{3/2} 2^t \rho_l^{2\sqrt{kt}+t} \text{  \; for } k\geq 21, \; 2 \leq t \leq (k-1)/9, \; l \in \mathbb{N}, \\
  q_{k,l,t} \leq& 2^{-1.52-4t}\frac{\rho_l^{6t}}{2^t-\rho_l^t}k + \rho_l^{3t}2^{-3.55-\frac{4k}{9}-2t}k^{15/4} + \rho_l^{5t}2^{1.75-\frac{k}{4}-3t}k\\
  &\text{ \; for } k\geq 122, \; t\geq k/9, \; l \in \mathbb{N},
\end{align*}
where $\tilde{p}_l$ is the $l$-th odd prime and $\rho_l =1+\frac{1}{\tilde{p}_{l+1}}$.

\section{Preliminaries}

\subsection{The Miller-Rabin Test}
The Miller-Rabin test, also referred to as the strong probable prime test, is a commonly used primality test. It exploits the following theorem:
\begin{theorem}
Let $p$ be a prime and write $p-1=2^\kappa q$, with $q$ odd. Then
\begin{equation} \label{MR-Theorem}
    \text{ either } a^q \equiv 1 \bmod p \text{ or } a^{2^iq} \equiv -1 \bmod p \text{ for $0 \leq i < \kappa $. }
\end{equation}
\end{theorem}
Similar to the strong Lucas test, the Miller-Rabin test involves checking property (\ref{MR-Theorem}) for multiple bases $a$. If a witness $a$ is found for which property (\ref{MR-Theorem}) does not hold, it indicates that $n$ is composite. On the other hand, if property (\ref{MR-Theorem}) holds for multiple bases, then $n$ is highly likely to be prime.

Composite numbers that satisfy condition (\ref{MR-Theorem}) are called \textit{strong pseudoprimes} with respect to the base $a$.
The following theorem, independently proven by Rabin \cite{Rabin} and Monier \cite{Monier} in 1980, provides an upper bound for the probability of this test giving an incorrect answer.
\begin{theorem}[The Rabin-Monier Theorem \cite{Rabin}, \cite{Monier}]\label{Rabin-Monier}
Let $n\neq 9$ be an odd composite integer. Let $S(n)$ denote the number of all bases $a$ relatively prime to $n$ such that $0 < a <n$ that make $n$ is a strong pseudoprime. We have $$S(n) \leq \frac{1}{4}\varphi(n),$$
where $\varphi$ is the Euler function.
\end{theorem}
We can directly conclude that this test has a worst case error probability of $1/4$ as  $\varphi(n)$ is bounded by $n$. The first known result that took advantage of the fact that on most  composite numbers the primality test has much smaller error probabilities than indicated by the worst case behavior was shown by Damgård, Landrock and Pomerance \cite{DamEtAl}. They considered an algorithm that repeatedly chooses random odd $k$-bit numbers, subjects each number to $t$ iterations of the Miller-Rabin test with randomly chosen bases, and outputs the first number found that passes all $t$ consecutive tests. Let $p_{k,t}$ be the probability that this algorithm falsely outputs a composite. They obtained numerical upper bounds for $p_{k,t}$ for various choices of $k, t$ and obtained an upper bound for $p_{k,t}$ for certain infinite classes of $k, t$. These bounds, which are formulated in the next theorem are still the best bounds we have for this primality test.

\begin{theorem}[Damgård, Landrock, Pomerance \cite{DamEtAl}]\label{results-dametal}
Let $k\geq 2$ and $t$ be integers. Then
\begin{enumerate}[(i)]
    \item  $p_{k,1}< k^2 4^{2-\sqrt{k}}$ for $k \geq 2$,\label{thrm:results1}
    \item $p_{k,t}<k ^{3/2} \frac{2^t}{\sqrt{t}}4^{2-\sqrt{tk}}$ for $k \geq 21, 3 \leq t \leq k/9$ or $k \geq 88, t=2,$\label{thrm:results2}
\item $p_{k,t}< \frac{7}{20}k2^{-5t}+\frac{1}{7}k^{15/4}2^{-k/2-2t}+12k2^{-k/4-3t}$ for $k\geq 21$ and $t\geq k/9$,\label{thrm:results3}
\item $p_{k,t}< \frac{1}{7}k^{15/4}2^{-k/2-2t}$ for $k\geq 21$ and $t\geq k/4.$ \label{thrm:results4}
\end{enumerate}
\end{theorem}
For specific large values of $k$, the paper has even better results. For example, they showed that $p_{500,1} < 4^{-28}$. Thus, if a randomly chosen odd 500-bit number passes just one iteration of a random Miller-Rabin test, the probability of it being composite is vanishingly small. Therefore, in most practical applications, such numbers can safely be accepted as ``prime''.

\subsection{Strong Lucas pseudoprimes}\label{stronglucas}
For the remainder of the paper, let $D$ be a fixed integer. The strong Lucas test is based on Theorem \ref{defslpsp}, which states a congruence condition that holds for all primes. Unfortunately there exist composites that satisfy congruence (\ref{thrm-slpsp}) for specific bases $(P,Q)$, while might failing the congruence for many other bases. Such odd composite numbers $n$ that are relatively prime to $2QD$ and satisfy the congruence condition are called \textit{strong Lucas pseudoprimes with respect to $P$ and $Q$}, denoted as $slpsp(P,Q)$. We define $SL(D,n)$ as the number of pairs $(P,Q)$ with ${0 \leq P,Q <n}$, $\gcd(Q,n)=1$, $P^2-4Q \equiv D \bmod n$, such that $n$  is a $slpsp(P,Q)$.

Arnault \cite{Rabin-Mon-Lucas} proved the following result for an integer $n$ with $\gcd(n,2D)=1$ on how many pairs $(P,Q)$ with $0 \leq P,Q <n$, $\gcd(Q,n)=1$, ${P^2-4Q \equiv D \bmod n}$ exist that make $n$ a $slpsp(P,Q)$.
\begin{theorem}[Arnault \cite{Rabin-Mon-Lucas}]\label{SL(D,n)}
Let $D$ be an integer and $n=p_1^{r_1}\cdot\ldots \cdot p_s^{r_s}$ be the prime decomposition of an integer $n\geq 2$ relatively prime to $2D$. Put
\begin{equation*}
 \begin{cases} n- \epsilon(n)=2^\kappa q 
    \\ p_i-\epsilon(p_i)=2^{k_i}q_i \text{ for } 1\leq i \leq s \end{cases} \text{ with } q, q_i \text{ odd },
    \end{equation*}
ordering the $p_i$'s such that $k_1 \leq \ldots \leq k_s$. The number of pairs $(P, Q)$ with ${0 \leq P, Q < n}$, $\gcd(Q,n)=1$, $P^2-4Q \equiv D \bmod n$ and such that $n$ is an \textnormal{slpsp}$(P,Q)$ is expressed by the formula
\begin{equation} \label{SL-formula}
    SL(D,n)= \prod_{i=1}^s (\gcd(q,q_i)-1) + \sum_{j=0}^{k_1-1}2^{js}\prod_{i=1}^s\gcd(q,q_i).
\end{equation}
If $n$ is not relatively prime to $2D$, we set $SL(D,n)=0$. 
\end{theorem}

We can define following function, which serves as a variant of $\varphi.$

\begin{definition}[Arnualt \cite{Rabin-Mon-Lucas}]
Let $D$ be an integer. The following number-theoretic function is defined only on integers relatively prime to $2D$:
\begin{equation*}
    \begin{cases}\varphi_D(p^r)=p^{r-1}(p-\epsilon(p)) \textnormal{ for any prime } p \nmid 2D \textnormal{ and } r\in\mathbb{N}
    \\ \varphi_D(p_1 p_2)= \varphi_D(p_1) \varphi_D(p_2) \text{ if } \gcd(p_1,p_2)=1. \end{cases}
\end{equation*}
\end{definition}
The following theorem could be seen as an analogue to Theorem \ref{Rabin-Monier} using $\varphi_D.$

\begin{theorem}[Arnault \cite{Rabin-Mon-Lucas}]\label{rabin-monier for Lucas}
Let $n$ be an odd composite integer relatively prime to $D$, then
\begin{equation*}
    SL(D,n) \leq \frac{\varphi_D(n)}{4}.
\end{equation*}
\end{theorem}

For the Miller-Rabin test, we can directly conclude that $S(n)< n/4$ since $\varphi(n)<n$. However, in contrast, Lemma \ref{lemmafinitelymany} shows that there are infinitely many $n$ for which $\varphi_D(n)$ is not bounded by $n$. Therefore, Theorem \ref{rabin-monier for Lucas} is not as relevant as Theorem \ref{Rabin-Monier} in this context. Nonetheless, we have the following useful result.

\begin{theorem}[Arnault \cite{Rabin-Mon-Lucas}]\label{Rabin-Monier-Lucas-using-n}
Let $D$ be an integer and $n\neq 9$ a composite integer relatively prime to $2D$. For every integer $D$, we have
$$
SL(D,n) \leq \frac{4n}{15},
$$
except if $n$ is the product if $n=(2^{k_1}q_1-1)(2^{k_1}q_1+1)$ of twin primes with $q_1$ odd and such that the Jacobi symbols satisfy $\epsilon(2^{k_1}q_1-1)=-1, \; \epsilon(2^{k_1}q_1+1)=1$. In this case we have $SL(D,n) \leq n/2.$
\end{theorem}
This theorem implies that for an odd composite integer not a product of twin-primes, at most $4/15$-th of the bases declare the integer as prime. Exclusing twin-prime products is not significant restriction. For $\epsilon(n)=-1$, where $n=p(p+2)$, the decomposition ${n-\epsilon(n)=(p+1)^2}$ can easily be detected using Newton's method for square roots before running the expensive primality test. Similarly, for $\epsilon(n)=1$, Newton's method can still be applied as $n-\epsilon(n)$ is almost a square.

\subsection{Some Lemmas and Corollaries}
In this section we establish lemmas that we will use in later the proofs.
Let $n$ be an odd integer an $\alpha_D(n)=\frac{SL(D,n)}{\varphi_D(n)}$. Thus, by Theorem \ref{rabin-monier for Lucas} we have $\alpha_D(n) \leq 1/4$ for odd composite $n$. 

Let $n-\epsilon(n)=2^\kappa q$, with $q$ odd. Also let $n=p_1^{r_1} \cdot \ldots \cdot p_s^{r_s}$ be the prime decomposition of an integer relatively prime to $2D$, ordering the $p_i$'s such that $k_1 \leq \ldots \leq k_s$ in the decomposition $p_i - \epsilon(p_i)=2^{k_i}q_i$, where $q_i$ is odd. This implies that $k_1$ is the largest integer such that $2^{k_1} \mid p_i -\epsilon(p_i)$ for all $i=1,2,\ldots, s$. Let $\omega(n)$ denote the number of distinct prime factors of $n$ and let $\Omega(n)$ denote the number of prime factors of $n$ counted with multiplicity. Thus, $\omega(n)=s$ and $\Omega(n)=\sum_{i=1}^s r_i$. We shall always let $p$ denote a prime number.

\begin{lemma}[Suwa \cite{suwa2012some}]\label{ki divides k}
Let $n$ be an odd integer $>1$. And let $\kappa= \nu_2(n-\epsilon(n))$ and $k_1=min_{p\mid n}\nu_2(p-\epsilon(p))$. Then we have $\kappa \geq k_1.$ Furthermore, equality holds if and only if the number of prime $p$ factors with odd exponent such that $\nu_2(p-\epsilon(p))=k_1$ is odd.
\end{lemma}

\begin{lemma}\label{sum_k_1}
Let $m, s \in \mathbb{N}$. Then \begin{equation*}
    \Bigg( 1+\sum_{j=0}^{m}2^{js} \Bigg) \leq 2^{m s+1}.
    \end{equation*}
\end{lemma}
The details of the proofs are omitted since this can be shown by induction on $m$.

\begin{lemma}\label{alpha_D}
If $n=p_1^{r_1}\cdot\ldots\cdot p_s^{r_s}>1$ is relatively prime to $2D$, then
\begin{align*}
    \alpha_D(n) &\leq 2^{1-s} \prod_{i=1}^sp^{1-r_i}\frac{\gcd(p-\epsilon(p), n - \epsilon(n))}{p-\epsilon(p)} \\
    &\leq 2^{1-\Omega(n)} \prod_{i=1}^s \frac{\gcd(p-\epsilon(p), n - \epsilon(n))}{p-\epsilon(p)}.
\end{align*}
\end{lemma}
\begin{proof}
We see that the identity $\sum_{i=1}^s (r_i -1) = \Omega(n)-s$ trivially holds. 
Thus, $$2^{(1-s)}=2^{1-\Omega(n)+\sum_{i=1}^s (r_i -1)}=2^{1-\Omega(n)} \prod_{i=1}^s2^{r_i -1}.$$ Using the fact that $\frac{2}{p} \leq 1$ for every prime $p$ and $r_i \geq 1$ for all $i$, the second inequality follows by
\begin{equation*}
2^{1-s} \prod_{i=1}^sp^{1-r_i} = 2^{1-\Omega(n)} \prod_{i=1}^s \frac{2^{r_i -1}}{p^{r_i-1}} \leq
2^{1-\Omega(n)}  \prod_{i=1}^s {\Big(\frac{2}{p}\Big)}^{r_i -1} \leq 2^{1-\Omega(n)}.
\end{equation*}

For the first inequality we use Theorem \ref{SL(D,n)}, which implies that for $n$ such that $\gcd(n,2D)=1$ we have
\begin{equation}\label{SLbound}
    SL(D,n)\leq \Bigg(1+\sum_{j=0}^{k_1-1}2^{js}\Bigg) \prod_{i=1}^{s}(q,q_i).
\end{equation}

Using this upper bound and the definition of $\varphi_D(n)$, we get
\begin{equation*}
\begin{split}
\alpha_D(n)=\frac{SL(D,n)}{\varphi_D(n)}&\leq \Bigg(1+\sum_{j=0}^{k_1-1}2^{j s}\Bigg) \prod_{i=1}^s\frac{\gcd(q_i,q)}{p_i^{r_i-1}(p_i-\epsilon(p_i))} \\ &=\Bigg(1+\sum_{j=0}^{k_1-1}2^{js}\Bigg) \prod_{i=1}^s\frac{\gcd(p_i-\epsilon(p_i),q)}{p_i^{r_i-1}(p_i-\epsilon(p_i))}.
\end{split}
\end{equation*}
Since in the factorization $n - \epsilon(n)=2^\kappa q$ the two factors $2^\kappa$ and $q$ are coprime, we get 
\begin{equation*}
\prod_{i=1}^s \gcd(p_i-\epsilon(p_i),n-\epsilon(n))=
\prod_{i=1}^s\gcd(p_i-\epsilon(p_i),q)\gcd(2^{k_i},2^k). 
\end{equation*}
By Lemma \ref{ki divides k} we know that $k_1 \leq \kappa $, and according to the way we have defined the order of $k_1, k_2, \dots, k_s$, we get
\begin{align*}
    \prod_{i=1}^s \gcd(p_i-\epsilon_D(p_i),q)\gcd(2^{k_i},2^\kappa) \geq 2^{sk_1} \prod_{i=1}^s \gcd(p_i-\epsilon_D(p_i),q).
\end{align*}
Hence,
\begin{align*}
     \prod_{i=1}^s\gcd(p_i-\epsilon(p_i),q) \leq 2^{- s k_1 }\prod_{i=1}^s \gcd(p_i-\epsilon(p_i), n -\epsilon(n)).
\end{align*}

Using Lemma \ref{sum_k_1}, we get
\begin{equation*}
\begin{split}
    \alpha_D(n) = \frac{SL(D,n)}{\varphi_D(n)} &\leq \Bigg( 1+\sum_{j=0}^{k_1-1}2^{j s} \Bigg) \prod_{i=1}^s \frac{\gcd(p_i-\epsilon(p_i),q)}{p_i^{r_i-1} (p_i-\epsilon(p_i))} \\
    &\leq \Bigg(1+\sum_{j=0}^{k_1-1}2^{j s} \Bigg) 2^{-k_1 s} \prod_{i=1}^s \frac{\gcd(p_i-\epsilon(p_i), n-\epsilon(n))}{p_i^{r_i-1} (p_i-\epsilon(p_i))} \\ &\leq \Big(2^{(k_1-1)s+1}\Big)  2^{-k_1 s} \prod_{i=1}^s \frac{\gcd(p_i-\epsilon(p_i), n-\epsilon(n))}{p_i^{r_i-1}(p_i-\epsilon(p_i))} \\ 
    &= 2^{1-s} \prod_{i=1}^s \frac{1}{p_i^{r_i-1}} \frac{\gcd(p_i-\epsilon(p_i), n-\epsilon(n))}{p_i-\epsilon(p_i)},
\end{split}
\end{equation*}
which proves the assertion.
\end{proof}

\begin{lemma}\label{part-frac}
Let $t\in \mathbb{R}$ with $t\geq 1$. Then 
$$ \sum_{n=\lfloor t\rfloor +1 }^\infty \frac{1}{n(n-1)}=\frac{1}{\lfloor t\rfloor}<\frac{2}{t}.$$
\end{lemma}
\begin{proof}
\begin{equation*}
    \sum_{n=\lfloor t\rfloor +1 }^\infty \frac{1}{n(n-1)}= \lim_{k\to\infty}\sum_{n=\lfloor t\rfloor +1}^k \frac{1}{n-1}-\frac{1}{n}=\lim_{k\to\infty}\frac{1}{\lfloor t\rfloor}+\frac{1}{k}=\frac{1}{\lfloor t\rfloor}<\frac{2}{t},
\end{equation*}
where we use the partial fractal decomposition of $\frac{1}{n(n-1)}$ and the fact that $\sum_n \Big(\frac{1}{n-1}-\frac{1}{n}\Big)$ is a telescope sum.
\end{proof}
The following lemma will also be frequently used.
\begin{lemma}[Damgård, Landrock, Pomerance \cite{DamEtAl}]\label{prop:eq:3}
For all $k,t,j\in \mathbb{N}$ we have
$$
2 \sqrt{tk} - \sqrt{\frac{t}{k-1}} \leq 2 \sqrt{t(k-1)} \leq jt + \frac{k-1}{j}.
$$
\end{lemma}
 
\subsection{A simple estimate}

Let us define the following set of integers that will be important in our analysis: ${C_{m,D}=\{n \in \mathbb{N} : \gcd(n,2D)=1, n \text{ composite and }\alpha_D(n) > 2^{-m}\}}$. By Theorem \ref{rabin-monier for Lucas} we already know that $\alpha_D(n) \leq \frac{1}{4}$. Hence, we have  $C_{1,D}= C_{2,D}=\emptyset$. We will focusing on classifying the set $C_{3,D}$ in Theorem \ref{comprise_C_3}.
Let $M_k$ denote the set of odd $k$-bit integers. For $k\geq 2$, the cardinality of $M_k$  is  $\lvert M_k \rvert = 2^{k-2}$. We are interested in determining the proportion of odd integers in $M_k$ that also belong to the set $C_{m,D}$.

\begin{theorem}\label{Frac_Cm_Estimate}
If $m,k$ are positive integers with $m+1 \leq 2 \sqrt{k-1}$, then $$\frac{\lvert C_{m,D} \cap M_k \rvert}{\lvert M_k \rvert}< 8 \sum_{j=2}^m 2^{m-j-\frac{k-1}{j}}.$$
\end{theorem}
\begin{proof}
Lemma \ref{alpha_D} with $n\in C_{m,D}$ implies that $\Omega(n)\leq m$. Now let ${N_D(m,k,j)=\{n \in C_{m,D} \cap M_k: \Omega(n)=j\}.}$ Thus
$$\lvert C_{m,D} \cap M_k \rvert = \sum_{j=2}^m \lvert N_D(m,k,j)\rvert.$$
Suppose $n \in N_D(m,k,j)$, where $2 \leq j \leq m$. Let $p$ denote the largest prime factor of $n$. Since $2^{k-1} < n < 2^k$, we have $p>2^{(k-1)/j}$. Let $d_D(p,n)=\frac{p-\epsilon(p)}{\gcd(p-\epsilon(p),n-\epsilon(n))}$. From Lemma \ref{alpha_D} and the definition of $C_{m,D}$ we have 
$$2^m > \frac{1}{\alpha_{D(n)}}\geq 2^{\Omega(n)-1}d_D(p,n)=2^{j-1}d_D(p,n),$$ so that $d_D(p,n)<2^{m+1-j}$. \\
Given $p,d$, where $p$ is a prime with the property that $p>2^{(k-1)/j}$ and $d$ is such that  $d \mid p-\epsilon(p)$ and $d<2^{m+1-j}$, we want to get an upper bound on how many $n\in N_D(m,k,j)$ exist that have largest prime factor $p$ such that and $d_D(p,n)=d$. Let  $S_{D,k,d,p}=\{n\in M_k : p \mid n, d_D(p,n)=d,  n \text{ composite}\}.$
The size of the set $S_{D,k,d,p}$ is at most the number of solutions of the system 
\begin{eqnarray*}
    n \equiv 0 \bmod p, &n\equiv \pm1 \bmod \frac{p-\epsilon(p)}{d}, & p<n<2^k,
\end{eqnarray*}
i.e. at most the set ${R_{D,k,d,p}=\{n \in \mathbb{Z} : n \equiv 0 \bmod p, n \equiv \pm 1 \bmod \frac{p-\epsilon(p)}{d}, p < n <2^k\}}$, which by the Chinese Remainder Theorem has less than $\frac{2^{k}d}{p(p-\epsilon(p))}$ elements.

If $S_{D,k,d,p} \neq \emptyset$, then there exists an $n \in S_{D,k,d,p}$ with ${\gcd(n-\epsilon(n),p-\epsilon(p))=(p-\epsilon(p))/d}$. Now let us look at the parity of $(p-\epsilon(p))/d$. Since both $p$ and $n$ are odd, $(p-\epsilon(p))/d=\gcd(p-\epsilon(p), n-\epsilon(n))$ must be even, we only need to consider those $p$ and $d$ that make $(p-\epsilon(p))/d$ even.
We conclude that 
\begin{align*}
    \lvert N_D(m,k,j) \rvert &\leq \sum_{p>2^{(k-1)/j}} \sum\limits_{\substack{d \mid p-\epsilon(p) \\ d<2^{m+1-j} \\ \frac{p-\epsilon(p)}{d} \in 2\mathbb{Z}} }\frac{2^k d}{p(p-\epsilon(p))}\\
    &=2^k \sum_{d<2^{m+1-j}} \sum\limits_{\substack{p>2^{(k-1)/j} \\ d \mid p-\epsilon(p) \\ \frac{p-\epsilon(p)}{d} \in 2\mathbb{Z}}} \frac{d}{p(p-\epsilon(p))}.
\end{align*}
Now, for the inner sum we have
\begin{align*}
&\sum\limits_{\substack{p>2^{(k-1)/j} \\ d \mid p-\epsilon(p) \\ \frac{p-\epsilon(p)}{d} \in 2\mathbb{Z}}}\frac{d}{p(p-\epsilon(p))} 
&<\sum_{2ud>2^{\frac{k-1}{j}}-\epsilon(p)}\frac{d}{(2ud+\epsilon(p))2ud}\\
&=\frac{1}{4d}\sum_{2ud>2^{\frac{k-1}{j}}-\epsilon(p)}\frac{1}{(u+\frac{\epsilon(p)}{2d})u} 
& \leq \frac{1}{4d} \sum_{2ud>2^{\frac{k-1}{j}}-\epsilon(p)} \frac{1}{u(u-\frac{1}{2d})} \\
&\leq \frac{1}{4d} \sum_{u>\frac{2^{\frac{k-1}{j}}-\epsilon(p)}{2d}} \frac{1}{u(u-1)}
&< \frac{1}{4d}\frac{2}{\frac{2^{\frac{k-1}{j}}-\epsilon(p)}{2d}}=\frac{1}{2^{\frac{k-1}{j}}-\epsilon(p)},
\end{align*}
where the last inequality follows from Lemma \ref{part-frac}. Using this estimate we get 
\begin{align*}
    \lvert N_D(m,k,j) \rvert \leq 2^k \sum_{d<2^{m+1-j}}\frac{1}{2^{\frac{k-1}{j}}-\epsilon(p)}= 2^{k}\frac{2^{m+1-j}-1}{2^{\frac{k-1}{j}}-\epsilon(p)}.
\end{align*}
Lemma \ref{prop:eq:3} with $t=1$ and our hypothesis that $m+1\leq 2 \sqrt{k-1}$  yields  ${m+1 \leq j +(k-1)/j}$. Thus,
\begin{equation*}
    \frac{2^{m+1-j}-1}{2^{\frac{k-1}{j}}-\epsilon(p)} \leq \frac{2^{m+1-j}-1}{2^{\frac{k-1}{j}}-1} \leq \frac{2^{m+1-j}}{2^{\frac{k-1}{j}}}= 2^{m-j-\frac{k-1}{j}+1}.  
\end{equation*}

Therefore, $N_D(m,k,j) \rvert \leq 2^{k+m-j-\frac{k-1}{j}+1}$. Combining everything and using the fact that $\lvert M_k \rvert =2^{k-2}$ yields
\begin{equation*}
\frac{\lvert C_{m,D} \cap M_k \rvert}{\lvert M_k \rvert} = \frac{\sum_{j=2}^m \lvert N_D(m,k,j)\rvert}{2^{k-2}} \leq  8 \sum_{j=2}^m 2^{m-j-\frac{k-1}{j}}.
\end{equation*}
\end{proof}

\subsection{The average case error probability}
We apply techniques similar to those used in \cite{DamEtAl}, with appropriate modifications for the strong Lucas test, to obtain average case error estimates. Let ${\overline{\alpha}_D(n)}=\frac{SL(D,n)}{n-\epsilon(n)-1}$ denote the fraction of pairs $(P,Q)$ for which the strong Lucas test is positive. We defin $X$ as the event that an integer $n$ declared as \textit{probable prime} by the strong Lucas test is composite, and $Y_t$ as the event that the uniformly at random chosen integer $n$ from $M_k$ passes $t$ consecutive rounds of the strong Lucas test with uniformly chosen bases $(P,Q)$. We also use $\pi(x)$ to denote the prime counting function up to $x$ and $\sum^{'}$ to denote the sum over composite integers. 
Using the law of conditional probability, we have 
\begin{align}\label{probLuc}
    q_{k,t}&= \mathbb{P}[X \mid Z_t]= \frac{\mathbb{P}[X \cap Z_t]}{\mathbb{P}[Z_t]} = \frac{\sum'_{n\in M_k}\overline{\alpha}_D(n)^t}{\sum_{n\in M_k}\overline{\alpha}_D(n)^t}\notag\\
    & \leq \frac{\sum'_{n\in M_k}\overline{\alpha}_D(n)^t}{\sum_{p\in M_k}\overline{\alpha}_D(p)^t} =\frac{\sum'_{n\in M_k}\overline{\alpha}_D(n)^t}{\pi(2^k)-\pi(2^{k-1})},
\end{align}
where $p$ is prime.

To obtain an upper bound for $q_{k,t}$, we need to find an upper bound for the final sum in (\ref{probLuc}) and a lower bound for $\pi(2^k)-\pi(2^{k-1})$. The latter quantity can be bounded using the following result:
\begin{proposition}[Damgård, Landrock, Pomerance \cite{DamEtAl}]\label{k-bit-prime-approx}
For an integer $k\geq 21$ , we have
\begin{equation}
    \pi(2^k)-\pi(2^{k-1}) > (0.71867)\frac{2^k}{k}.
\end{equation}
\end{proposition}

To proceed, we aim to find an upper bound for the sum $\sum'_{n\in M_k}\overline{\alpha}_D(n)^t=\sum_{m=2}^\infty \sum_{n\in M_k \cap C_{m,D}\setminus C_{m-1,D}}\overline{\alpha}_D(n)^t.$ However, it is not clear how to bound $\overline{\alpha}_D(n)$ directly. Theorem \ref{Frac_Cm_Estimate} provides a way to upper bound $\mid C_{m,D} \cap M_k \mid$. However, we still need a method to bound $\overline{\alpha}_D(n)$ using $\alpha_D(n)$. If we can achieve this, we can utilize the property that for $n\in C_{m,D} \setminus C_{m-1,D}$, 
we have $2^{-m} < \alpha_D(n) \leq 2^{-(m-1)}$. Thus, our challenge lies in bounding $\overline{\alpha}_D(n)$ using $\alpha_D(n)$. We tackle this problem by establishing two different procedures: for the general case and another for a scenario that involves trial division by small primes before conducting the more computationally expensive strong Lucas tests. The latter procedure yields improvements over the general procedure, comparable to the results obtained for the Miller-Rabin test in \cite{DamEtAl}.

For $n=\prod_{i=1}^s p_i^{r_i}$ we  have
\begin{align}\label{tightboundphi}
    \varphi_D(n)=\prod_{i=1}^s p_i^{r_i-1}(p_i-\epsilon(p_i)) \leq \prod_{i=1}^s (p_i^{r_i}+p_i^{r_i-1}).
\end{align}
Let's investigate whether or not this is an overestimate, that is, whether integers $n= \prod_{i=1}^n p_i^{r_i}$ actually exist with $\epsilon(p_i)=-1$ for all prime factors $p_i$ of $n$. The following theorem provides the answer to this question. 
\begin{theorem}[Ireland \cite{ireland1990classical}] \label{lemmafinitelymany}
Let $D$ be a non-square integer. Then there exists infinitely many primes $p$ for which $D$ is a quadratic non-residue.
\end{theorem}

Hence, there exists an infinite number of integers that achieve the bound stated in (\ref{tightboundphi}). This demonstrates that the bound is tight and cannot be weakened in general.
\begin{proposition}\label{infinitely}
For every $D$ there are infinitely many integers of the form
${n=\prod_{i=1}^s p_i^{r_i}}$ co-prime to $2D$ with $\varphi_D(n)=\prod_{i=1}^s (p_i^{r_i}+p_i^{r_i-1})$.
\end{proposition}

\section{Explicit bounds for \texorpdfstring{$q_{k,t}$}{qkt}}\label{firstapproach}
In this section we establish explicit bounds for $q_{k,t}$.

\subsection{A bound for \texorpdfstring{$\overline{\alpha}_D(n)$}{alpha}}
Proposition \ref{infinitely} demonstrates that in general, $\varphi_D(n)$ is not necessarily bounded by $n$. Consequently, we cannot directly conclude that $\overline{\alpha}_D(n) \leq \alpha_D(n)$. However, in order to continue our analysis, we can establish a relationship between the two functions.
\begin{theorem}[Akbary, Friggstad \cite{explicit-akbary}]\label{akbary}
\begin{equation*}
\frac{n}{\varphi(n)}\leq  1.07e^{\gamma}\log(\log(n)) \text{ \; \; \; for } n \geq 2^{78},
\end{equation*}
where $\gamma$ is the Euler-Mascheroni constant:
\begin{equation*}
   \gamma = \lim_{n \to \infty} \Bigg( \sum_{k=1}^n \frac{1}{k}- \ln(n) \Bigg) <0.58.
\end{equation*}
\end{theorem}
Using this result we obtain an explicit upper bound for  $\varphi_D$.

\begin{lemma}\label{bounding-phi_D}
For integers $k \geq 78$ and $n \in M_k$ we have
\begin{equation*}
\varphi_D(n)<2n\log(k).
\end{equation*}
\end{lemma}

\begin{proof}
\begin{align} \label{boundingphi}
    \varphi_D(n)\leq 
 n \prod_{i=1}^s \Bigg( 1+\frac{1}{p_i} \Bigg)
  \leq n \prod_{i=1}^s \Bigg( 1 + \frac{1}{p_i-1} \Bigg) = \frac{n}{\prod_{i=1}^s \Big( 1-\frac{1}{p_i}\Big)}.
\end{align}
We realize that $ \prod_{i=1}^s \Big( 1- \frac{1}{p_i} \Big) = \frac{\varphi(n)}{n}$. Using this in (\ref{boundingphi}) we obtain $\varphi_D(n) \leq n\frac{n}{\varphi(n)}.$ For $k \geq 78$ and $n \in M_k$ we obtain by Theorem \ref{akbary} that
\begin{align*}
    \varphi_D(n) \leq n  \frac{n}{\varphi(n)} < n 1.07e^{\gamma}\log(\log(n))< 2n\log(\log(2^k))<2n\log(k), 
\end{align*}
as $1.07e^{\gamma}<2$, which proves the claim.
\end{proof}
Therefore, we immediately get the following estimate for $\overline{\alpha}_D$.
\begin{corollary}\label{bound-alpha-1}
For $k \geq 78$ and $n\in M_k$ we have
\begin{equation*}
    \overline{\alpha}_D(n) \leq 2\log(k) \alpha_D(n).
\end{equation*}
\end{corollary}

\subsection{An intermediate result}
Corollary \ref{bound-alpha-1} and Theorem \ref{Frac_Cm_Estimate} allow us to proceed with our analysis.

\begin{proposition}\label{prop_alpha_comp}
For any integers $k, M, t$ with $3\leq M \leq 2\sqrt{k-1}-1, t\geq 1$ and $k \geq 78$ we have
$$
\sum{'}_{n\in M_k}\overline{\alpha}_D(n)^t \leq 2^{k-2+t(1-M)}\log^t(k) + 2^{k+1+2t}\log^t(k)\sum_{j=2}^M \sum_{m=j}^M 2^{m(1-t)-j-\frac{k-1}{j}}.
$$
\end{proposition}

\begin{proof}
Note that our hypothesis implies $k \geq 5$. We know that $C_{1,D} \cap M_k=\emptyset$. Thus, by Corollary \ref{bound-alpha-1} we have
\begin{align*}
    \sum{'}_{n \in M_k}\overline{\alpha}_D(n)^t &= \sum_{m=2}^\infty \sum_{n\in M_k\cap C_{m,D} \setminus C_{m-1}} \overline{\alpha}_D(n)^t \\
    &\leq \sum_{m=2}^\infty \sum_{n\in M_k\cap C_{m,D} \setminus C_{m-1,D}}\big(2\log(k)\alpha_D(n)\big)^t.
\end{align*}
Since $ n \in C_{m,D} \setminus C_{m-1,D}$ we have that $2^{-m} < \alpha_D(n) \leq 2^{-(m-1)}$. Hence, we get

\begin{align*}
     \sum{'}_{n \in M_k}\overline{\alpha}_D(n)^t &< \log^t(k) \sum_{m=2}^\infty 2^{t-(m-1)t}\lvert M_k\cap C_{m,D} \setminus C_{m-1,D} \rvert\\
   &  \leq  \log^t(k)\Big(2^{t(1-M)} \lvert M_k \setminus C_{M,D}\rvert+ \sum_{m=2}^M 2^{(2-m)t} \lvert M_k \cap C_{m,D} \rvert \Big). 
\end{align*}

Using Theorem \ref{Frac_Cm_Estimate} in the above estimate we have
\begin{align*}
    \sum{'}_{n \in M_k}\overline{\alpha}_D(n)^t &\leq \log^t(k)\Big(2^{k-2+t(1-M)}+ 2^{k+1+2t}\sum_{m=2}^M \sum_{j=2}^m 2^{m(1-t)-j-\frac{k-1}{j}} \Big) \\
 &= \log^t(k)\Big(2^{k-2+t(1-M)} + 2^{k+1+2t}\sum_{j=2}^M \sum_{m=j}^M 2^{m(1-t)-j-\frac{k-1}{j}}\Big).
\end{align*} \end{proof}

\subsection{An estimate for \texorpdfstring{$q_{k,1}$}{q k,1}}
We now derive the first numerical upper bound for $q_{k,t}$ when $t=1$.

\begin{theorem}\label{k, 1 theorem}
For $k \geq 2,$ we have $q_{k,1} < \log(k)k^24^{2.3-\sqrt{k}}.$
\end{theorem}
\begin{proof}
We use Proposition \ref{prop_alpha_comp} with $t=1$ and $k\geq 78$ and let $M$ be an integer with $3\leq M \leq 2\sqrt{k-1}-1$ and get
\begin{align}
    \sum{'}_{n \in M_k}\overline{\alpha}_D(n) &\leq 
   \log(k)\Big(2^{k-1-M} +2^{k+3}\sum_{j=2}^M(M+1-j) 2^{-j-\frac{k-1}{j}}\Big) \notag \\
   &\leq \log(k) \Big( 2^{k-1-M}+2^{k+3-2\sqrt{k-1}}\sum_{j=2}^M(M+1-j)\Big), \label{eq-needed-for-further-case}
\end{align}
where we used Lemma \ref{prop:eq:3} to bound $2^{-jt -\frac{k-1}{j}}$. We bound the sum ${\sum_{j=2}^M(M+1-j)=M(M-1)/2}$ and let $M=\lfloor 2\sqrt{k-1}-1 \rfloor$ which yields
\begin{align}\label{prop:eq:2}
    \sum{'}_{n \in M_k}\overline{\alpha}_D(n) &\leq \log(k) \Big(2^{k-1-M}+2^{k+2-2\sqrt{k-1}}M(M-1) \Big)\nonumber\\ 
    &< \log(k) \Big(2^{k+1-2\sqrt{k-1}}(1+2(4(k-1)-6\sqrt{k-1}+2) \Big)\nonumber \\
    &<\log(k)k2^{k+4-2\sqrt{k-1}}.
\end{align}

We again use Lemma \ref{prop:eq:3} with $t=1$ in inequality (\ref{prop:eq:2}) for $k\geq 100$ and get 
\begin{equation}\label{prop:eq:4}
    \sum{'}_{n \in M_k}\overline{\alpha}_D(n)<\log(k)k2^{4+\frac{1}{\sqrt{99}}+k-2\sqrt{k}}.
\end{equation}

As $\frac{2^{4+\frac{1}{\sqrt{99}}}}{0.71867}<4^{2.3}$ we get by  Proposition \ref{k-bit-prime-approx} and inequalities (\ref{prop:eq:4}) and (\ref{probLuc}) for $k\geq 100$ that

\begin{equation*}
    q_{k,1} \leq \frac{\sum{'}_{n \in M_k}\overline{\alpha}_D(n)}{\pi(2^k)-\pi(2^{k-1})} = \frac{\log(k)k^2\cdot 2^{4+\frac{1}{\sqrt{99}}-2\sqrt{k}}}{0.71867}<\log(k)k^24^{2.3-\sqrt{k}}.
\end{equation*}
But for $k\leq 101$ we have that $\log(k)k^24^{2.3-\sqrt{k}}>1$, so this upper bound is trivially true for $k \leq 101$. 
\end{proof}

\subsection{An estimate for \texorpdfstring{$q_{k,t}$}{q k,t}}
We will now consider the average case error estimate for a number that has passed $t$ consecutive rounds of the strong Lucas test with respect to randomly chosen bases, and obtain numerical bounds for $q_{k,t}$ when $t\geq 2.$

\begin{theorem}\label{k,t theorem}
For integers $k, t$  with $k\geq 78$, $3 \leq t \leq k/9$ or $k\geq 88$, $t=2$ we have
$$
q_{k,t}<\log^t(k)k^{3/2}\frac{2^t}{\sqrt{t}}4^{2.12-\sqrt{tk}}.
$$
\end{theorem}

\begin{proof}
Assume $k\geq 78$ and $t\geq 2$. Let us first estimate $\sum_{m=j}^M2^{m(1-t)}$. We do this by seeing that $${\sum_{m=j}^M2^{m(1-t)} = \sum_{m=0}^M2^{m(1-t)}-\sum_{m=0}^{j-1}2^{m(1-t)}=\frac{2^{1-t}(2^j-2^M)}{1-2^{1-t}}\leq\frac{2^{j(1-t)}}{1-2^{1-t}}},$$ as $j \leq M$. Using this estimate in Proposition \ref{prop_alpha_comp} we get that
\begin{align}\label{p_kt:eq:1}
    \sum{'}_{n\in M_k}\overline{\alpha}_D(n)^t &\leq
      2^{k-2+t(1-M)}\log^t(k)+
    \frac{2^{k+1+2t}}{1-2^{1-t}}\log^t(k)\sum_{j=2}^M 2^{-jt-\frac{k-1}{j}},
\end{align}
for any integer $M$ with $3 \leq M \leq 2\sqrt{k-1}-1$. By Lemma \ref{prop:eq:3} we have that
$$
jt+\frac{k-1}{j}\geq 2\sqrt{t(k-1)} \; \; \; \forall j,k >0.
$$

Furthermore, we choose $M= \big\lceil 2 \sqrt{(k-1)/t}+1 \big\rceil$. In order to use Proposition \ref{prop_alpha_comp}, we need to make sure that $3 \leq M \leq 2\sqrt{k-1}-1$.  Thus, for $3 \leq M$ to hold, we must restrict $t \leq k-1$ for $k>1$. For $k \geq 25$, we have  $M \leq 2\sqrt{k-1}-1$.


Our choice of $M$ implies that $M-1 < 2 \sqrt{(k-1)/t}+1< 2\sqrt{k/t}+1$. Since $t \leq k-1$, we have $1 \leq \sqrt{k/t}$, which yields $M-1 < 2^{1.6}\sqrt{k/t}$.
We also see that $M-1 \leq  2\sqrt{(k-1)/t}.$

Using our chosen value for $M$ and the inequalities established above in (\ref{p_kt:eq:1}), we can conclude that 
\begin{align*}
 \sum{'}_{n\in M_k}\overline{\alpha}_D(n)^t &\leq 2^{k-2+t(1-M)}\log^t(k)+ \frac{2^{k+1+2t}}{1-2^{1-t}}\log^t(k)(M-1)2^{-2\sqrt{t(k-1)}}\\
 &< 2^{k-2-2\sqrt{t(k-1)}}\log^t(k)+\frac{2^{k+2.6+2t}}{1-2^{1-t}}\log^t(k)\sqrt{\frac{k}{t}}2^{-2\sqrt{t(k-1)}}\\
 &=2^{k-2-2\sqrt{t(k-1)}}\log^t(k)\Bigg(1+2^{4.6}\frac{2^{2t}}{1-2^{1-t}}\sqrt{\frac{k}{t}} \Bigg).
\end{align*}
The function $f(k,t)=\frac{2^{2t}}{1-2^{1-t}}\sqrt{\frac{k}{t}}$ is a monotonically increasing function for all $t\geq 2$ and for all $k \geq 1$. Thus, we get with $k \geq 79$ and $t\geq 2$ that
$$
2^{4.6}\frac{2^{2t}}{1-2^{1-t}}\sqrt{\frac{k}{t}}\geq 2^{4.6}\frac{2^4}{1-2^{-1}}\sqrt{\frac{79}{2}} = 4877.38 > 4877.
$$
For $x>4877$ we have $1+x=x\big(\frac{1}{x}+1\big)<x\frac{4878}{4877}$, which yields
\begin{equation}\label{p_kt:eq:2}
\sum{'}_{n\in M_k}\overline{\alpha}_D(n)^t < 2^{k-2-2\sqrt{t(k-1)}}\log^t(k)\frac{4878}{4877}2^{4.6}\frac{2^{2t}}{1-2^{1-t}}\sqrt{\frac{k}{t}}. 
\end{equation}

We upper bound $2^{-2\sqrt{t(k-1)}}$ using Lemma \ref{prop:eq:3}.
For $t=2$ and $k\geq 88$, and using the fact that $2^{1+\sqrt{\frac{2}{k-1}}}$ is a monotonically decreasing function for all $k \geq 1$, we have
$$
\frac{2^{\sqrt{\frac{t}{k-1}}}}{1-2^{1-t}}=\frac{2^{\sqrt{\frac{2}{k-1}}}}{1-2^{1-2}}=2^{1+\sqrt{\frac{2}{k-1}}}< 2.222.
$$
For $3 \leq t \leq k/9$, we have
$$
\frac{2^{\sqrt{\frac{t}{k-1}}}}{1-2^{1-t}} \leq \frac{4}{3}2^{\frac{3}{26}}<1.7.
$$
In any case we have $\frac{2^{\sqrt{\frac{t}{k-1}}}}{1-2^{1-t}} < 2.222$. Putting these estimates in ($\ref{p_kt:eq:2}$), we get
\begin{align*}
\sum{'}_{n\in M_k}\overline{\alpha}_D(n)^t &< 2^{k-2\sqrt{tk}+2t}\log^t(k)\frac{4878}{4877}2^{2.6}\frac{2^{\sqrt{\frac{t}{k-1}}}}{1-2^{1-t}}\sqrt{\frac{k}{t}} \\
&< 2^{k-2\sqrt{tk}+2t}\log^t(k)\frac{4878}{4877}2^{2.6}2.222\sqrt{\frac{k}{t}}
\end{align*}
for all $3 \leq t \leq (k-1)/2$, $k\geq 79$ and for $t=2$, $k\geq 88$.
Now, using Proposition \ref{k-bit-prime-approx} and inequality (\ref{probLuc}), we obtain the desired result.
\end{proof}

The bounds obtained in Theorems \ref{k, 1 theorem} and \ref{k,t theorem} already show that for most choices of $k$ and $t$, the average case error estimates for the strong Lucas test are small enough to be used in practice. Yet, there is still room for improvement. For example, Theorem \ref{k, 1 theorem} implies $p_{k,1}>1$ for all $k\leq 101$. Despite the bounds in Theorem \ref{k,t theorem} always being less than 1, they are not as good as expected, especially for small choices of $t=2,3,4$ and small choices of $k$, such as $k\leq 200$.
In the next section, we will present even better bounds by focusing solely on integers that are not divisible by small primes.

\section{Improved average case error estimates}\label{secondapproach}
The estimates of $q_{k,t}$ established in Section \ref{firstapproach} provide satisfactory results but are weaker than the average case error estimates of the Miller-Rabin test in \cite{DamEtAl}. Therefore, we are interested in tightening our bounds. To achieve this, we propose a modification of the primality testing algorithm by incorporating a subroutine that performs trial division by the first $l$ odd primes before applying the strong Lucas test. In practice, this subroutine is already implemented in OpenSSL to accelerate prime generation, thus adding it usually does not incur additional running time. This modification enables us to establish a new bound for $\overline{\alpha}_D$, leading to improved average case error estimates.

\begin{remark}\label{remark1}
The primality testing function provided in OpenSSL Version 3.0\footnote[1]{See \url{https://github.com/openssl/openssl/blob/54a0d4ceb28d53f5b00a27fc5ca8ff8f0ddf9036/crypto/bn/bn_prime.c}} is called \verb|BN_is_prime_fasttest_ex| and is located in the \verb|bn_prime.c| file. Within this function, there is a call to \verb|calc_trial_divisions|, which calculates the optimal number of trial divisions for the Miller-Rabin test based on the bit-length $k$, ensuring the best speed-performance combination.
\begin{verbatim}
static int calc_trial_divisions(int k)
{
    if (k <= 512)
       return l = 63;
    else if (k<=1024)
        return l= 127;
    else if (k<=2048)
       return l = 383;
    else if  (k<= 4096)
       return l = 1023;
}
\end{verbatim}
Afterward, the function calls another function that invokes the Miller-Rabing testing with pseudo-random bases. The number of rounds also depends on the size of $k$.
\end{remark}

\subsection{An improved bound for \texorpdfstring{$\overline{\alpha}_D(n)$}{alphaD}}

For the rest of this section, let $\tilde{p_l}$ denote the $l$-th odd prime, and let $\rho_l =  1+ \frac{1}{\tilde{p}_{l+1}}$. Let $M_{k,l}$ represent the set of odd $k$-bit integers that are not divisible by the first odd $l$ primes.
We define $q_{k,l,t}$ as the probability that a composite integer, chosen uniformly at random from $M_{k,l}$, passes $t$ consecutive rounds of the strong Lucas test with randomly chosen bases $(P,Q).$

We will make use of the following two lemmas in our analysis.
\begin{lemma}\label{newbound}
Let $l, n\in \mathbb{N}$ and let $n$ be relatively prime to $2D$ and not divisible by all of the first $l$ odd primes. Then
\begin{equation*}
    \varphi_D(n) \leq \rho_l^{\omega(n)}  n,
\end{equation*}
which implies,
\begin{equation*}
    \overline{\alpha}_D(n) \leq \rho_l^{\omega(n)} \alpha_D(n).
\end{equation*}
\end{lemma}
\begin{proof}
For $n_1, n_2 \in \N$ with $\gcd(n_1,n_2)=1$, we have the relation
$$
\varphi_D(n_1n_2)=\varphi_D(n_1)\varphi_D(n_2).
$$
It is thus sufficient to only treat the case $n=p^r$. We have
$$
\frac{\varphi_D(p^r)}{p^r}=\frac{p^{r-1}(p-\epsilon(p))}{p^r}= 1-\frac{\epsilon(p)}{p}\leq 1+\frac{1}{p}.
$$
With $p \geq \tilde{p}_{l+1}$, the result follows directly.
\end{proof}

\begin{lemma}\label{omega,m}
Let $n\in C_{m,D} $. Then
\begin{equation*}
\omega(n) \leq m.
\end{equation*}
\end{lemma}
\begin{proof}
The result directly follows from Lemma \ref{alpha_D}, which states
\begin{align*}
\alpha_D(n) \leq 2^{1-{\omega(n)}} \prod_{i=1}^{\omega(n)} p^{1-r_i} \frac{\gcd(p-\epsilon(p), n - \epsilon(n))}{p-\epsilon(p)}
 \leq 2^{1-\omega(n)}.
\end{align*}
\end{proof}
We can now give a bound similar to Proposition \ref{prop_alpha_comp} for $\sum{'}_{n\in M_{k,l}}\overline{\alpha}_D(n)^t.$

\begin{proposition}\label{alpha-split-primes}
For any integers $k,t,M,l$ with $3 \leq M \leq 2\sqrt{k-1}-1$, we have
\begin{equation*}
    \sum{'}_{n\in M_{k,l}}\overline{\alpha}_D(n)^t \leq 2^{k-2+t}\sum_{m=M+1}^\infty  \rho_l^{mt} 2^{-mt} + 2^{k+1+t}\sum_{m=2}^M  \sum_{j=2}^m\rho_l^{mt} 2^{m(1-t)-j-\frac{k-1}{j}}.
\end{equation*}
\end{proposition}
\begin{proof}
The theorem follows by closely following the proof of Proposition \ref{prop_alpha_comp} while using Lemmas \ref{newbound} and \ref{omega,m}.
\begin{align*}
    \sum{'}_{n\in M_{k,l}}\overline{\alpha}_D(n)^t &= \sum_{m=2}^\infty \sum_{n \in M_{k} \cap C_{m,D} \setminus C_{m-1,D}} \overline{\alpha}_D(n)^t \leq \sum_{m=2}^\infty \sum_{n \in M_{k} \cap C_{m,D} \setminus C_{m-1,D}} \rho_l^{mt} 2^{-(m-1)t} \\
    &\leq \sum_{m=M+1}^\infty  \rho_l^{mt} 2^{-(m-1)t}\mid M_{k} \mid + \sum_{m=2}^M \rho_l^{mt} 2^{-(m-1)t} \mid M_{k} \cap C_{m,D} \mid\\
&\leq 2^{k-2+t}\sum_{m=M+1}^\infty  \rho_l^{mt} 2^{-mt} \label{ineqlast}+ 2^{k+1+t}\sum_{m=2}^M  \sum_{j=2}^m\rho_l^{mt} 2^{m(1-t)-j-\frac{k-1}{j}}\\
\end{align*}
\end{proof}

In this new bound for $\sum{'}_{n\in M_{k,l}}\overline{\alpha}_D(n)^t$, we successfully eliminated the term $\log(k)^t$. Additionally, if $l$ is chosen as discussed in Remark \ref{remark1}, $\rho_l$ becomes close to $1$, making it almost negligible in comparison to the dominant factors.
\subsection{An estimate for  \texorpdfstring{$q_{k,l,1}$ }{qk,l,1}}
For $t=1$, we need the following proposition to establish a new estimate.

\begin{proposition}\label{t=1,alphanew}
For any integers $k, M, l $ with $3 \leq M \leq 2\sqrt{k-1}-1$, we have
\begin{align*}
\sum{'}_{n\in M_{k,l}}\overline{\alpha}_D(n) \leq 2^{k-1-M}{\rho_l}^{M+1} + 2^{k-2\sqrt{k-1}+1}\rho_l^{M}M(M-1).
\end{align*}
\end{proposition}
\begin{proof}
Evaluating the first part of the sum in Proposition \ref{alpha-split-primes} with $t=1$ yields $\sum_{m=M+1}^\infty  \rho_l^{m} 2^{-m} = \frac{2^{-M}{\rho_l}^{M+1}}{2-\rho_l}\leq 2^{-M}{\rho_l}^{M+1}$.

For the second part of the sum, using  Lemma \ref{prop:eq:3} with $t=1$ and the condition  $m\leq M$, we conclude that
\begin{align*}
2^{k+2}\sum_{m=2}^M  \sum_{j=2}^m\rho_l^{m} 2^{-j-\frac{k-1}{j}} \leq 2^{k-2\sqrt{k-1}} \rho_l^{M}\sum_{j=2}^M \sum_{m=j}^M 1= 2^{k+1-2\sqrt{k-1}} \rho_l^{M}M(M-1).
\end{align*}
\end{proof}

\begin{theorem}\label{probability-estimate}
For $k\geq 2$ and $l\in \mathbb{N}$, we have
\begin{equation*}
  q_{k,l,1}< k^24^{1.8-\sqrt{k}}  \rho_l^{2\sqrt{k-1}-2}.
\end{equation*}
\end{theorem}

\begin{proof}
Using inequality (\ref{probLuc}) and Proposition \ref{t=1,alphanew} with $M=\lfloor 2\sqrt{k-1} -2 \rfloor$,  we get
\begin{align}\label{boundq-1-734}
q_{k,l,1}=\frac{\sum{'}_{n\in M_{k,l}}\overline{\alpha}_D(n)^t }{\pi(2^k)-\pi(2^{k-1})} \leq k^2 4^{1.73-\sqrt{k-1}}\rho_l^{2\sqrt{k-1}-1}.
\end{align}
Using Lemma \ref{prop:eq:3}, for $k\geq 53$, we have that
\begin{align*}
    4^{-\sqrt{k-1}}<4^{\frac{1}{4\sqrt{13}}-\sqrt{k}+}< 4^{0.07-\sqrt{k}}.
\end{align*}
Therefore, we get  
\begin{equation*}
  q_{k,l,1}< k^24^{1.8-\sqrt{k}}  \rho_l^{2\sqrt{k-1}-2},
\end{equation*}
which proves the theorem for $k\geq 53.$ The theorem is trivially true for $k\leq 52,$ as $k^24^{1.8-\sqrt{k}}  \rho_l^{2\sqrt{k-1}-2} \geq k^24^{1.8-\sqrt{k}}\geq 1$ for $k\geq 52$ and $l\geq 1$.
\end{proof}

Let's examine the bound for $q_{k,l,1}$ in Theorem \ref{probability-estimate} more closely. When the $l$-th prime is sufficiently large, $\rho_l$ is approximately equal to $1$. For instance, when $k=1024$ and $l=127$, we find that $\rho_l^{2\sqrt{k-1}-1} <1.09$.

\begin{corollary}\label{corr-not-divs-128}
Let $n$ be an odd integer which is not divisible by the first 127 odd primes. Then for all $k\geq 2$ we have $q_{k,127,1}< k^24^{1.729 - 0.998\sqrt{k-1}}.$ 
\end{corollary}
\begin{proof}
Using (\ref{boundq-1-734}) we have that
$$q_{k,127,1} \leq k^2 4^{-\sqrt{k-1}+1.73}\Big( \frac{728}{727} \Big)^{2\sqrt{k-1}-1}.$$
With $\big(\frac{728}{727}\big)^{2\sqrt{k-1}-1} \leq 4^{(2\sqrt{k-1}-1)0.001}$ we get 
\begin{equation*}
q_{k,127,1} \leq k^2 4^{1.73-\sqrt{k-1}+0.001(2\sqrt{k-1}-1)} \leq k^24^{1.729 - 0.998\sqrt{k-1}},
\end{equation*}
which proves the corollary.
\end{proof}
In Table \ref{probabilitiesqk1}, we compare the bounds of $p_{k,1}$ for the Miller-Rabin test, as stated in Theorem \ref{results-dametal} (\ref{thrm:results1}) of \cite{DamEtAl}, with the bounds of $q_{k,1}$ and $q_{k,l,1}$ for the strong Lucas test, as established in Theorem \ref{k, 1 theorem} and Theorem \ref{probability-estimate} respectively. For the bounds of $q_{k,l,1}$, we select the values of $l$ with respect to $k$ as defined in the \verb|calc_trial_division| function, as discussed in Remark \ref{remark1}.

\begin{table}[H]
    \centering
    \begin{tabular}{ l| l l  l }
	$k$ & $-\log_2p_{k,1}$ & $-\log_2q_{k,1}$ & -$\log_2q_{k,l,1}$ \\ \hline
	100 & 2 & -1 & 3 \\ 
	200 & 8 & 5 & 9 \\ 
	400 & 18& 15& 18 \\ 
	512 & 23 & 20 & 23 \\
	1024 & 40 & 36 & 40 \\ 
	2048 & 64 & 60 & 64 \\
	4096 & 100 & 96& 100 \\
    \end{tabular}
    \caption{Comparing the lower bound probabilities for $-\log_2(prob)$, where $prob=p_{k,1}, q_{k,1}, q_{k,l,1}$, and $l$ was chosen with respect to $k$ as discussed in Remark \ref{remark1}.}
    \label{probabilitiesqk1}
\end{table}

\subsection{An estimate for \texorpdfstring{$q_{k,l,t}$ }{q,k,l,t}}
In this section, we let $t\geq 2$ and  establish bounds for $q_{k,l,t}$ for various choices of $k,t$ and $l$.

\begin{corollary}\label{2^talpha-estimate}
Let $\rho_l = 1+ \frac{1}{\tilde{p}_{l+1}}$. Then
$$
2^t-\rho_l^t \geq \frac{1}{2}\rho_l^t.
$$
\end{corollary}
\begin{proof}
 $\rho_l \leq \frac{4}{3}<2$ implies in
$2^t-\rho_l^t \geq \rho_l^t\big(\frac{2}{\rho_l}-1\big) \geq \frac{1}{2}\rho_l^t.$
\end{proof}
We are now ready to prove the bound for $q_{k,l,t}.$
\begin{theorem}\label{q-k,t new}
For any integers $2 \leq t \leq (k-1)/9$, $k \geq 21, l\in \N $ we have
\begin{equation*}
    q_{k,l,t}\leq 4^{1.72-\sqrt{tk}}k^{3/2} 2^t \rho_l^{2\sqrt{kt}+t}.
\end{equation*}
\end{theorem}
\begin{proof}
By Proposition \ref{alpha-split-primes} we know that

\begin{equation}
\label{alpha-split-2}
    \sum{'}_{n\in M_{k,l}}\overline{\alpha}_D(n)^t \leq 2^{k-2+t}\sum_{m=M+1}^\infty  \rho_l^{mt} 2^{-mt} 
    +2^{k+1+t}\sum_{j=2}^M  \sum_{m=j}^M \rho_l^{mt} 2^{m(1-t)-j-\frac{k-1}{j}}.
\end{equation}
for any integer $2 \leq M \leq 2\sqrt{k-1}-1.$
Let us first look at the left hand side of the sum of (\ref{alpha-split-2}). Using Corollary \ref{2^talpha-estimate} we get that
\begin{align}\label{part1alpha}
2^{k-2+t}\sum_{m=M+1}^\infty  \rho_l ^{mt} 2^{-mt}  = 2^{k-2+t} \frac{2^{-Mt}\rho_l^{t(M-1)}}{2^t-\rho_l^t} \leq 2^{k-1-(M-1)t}\rho_l^{(M-2)t}.
\end{align}
Now let us look at the right hand side of the sum of (\ref{alpha-split-2}).
Using ${\sum_{m=j}^M 2^{m(1-t)} < \frac{2^{j(1-t)+t}}{2^t-2}}$, and $m\leq M$ we obtain
\begin{equation}\label{part2alpha}
2^{k+1+t}\sum_{j=2}^M  \sum_{m=j}^M\rho_l^{mt} 2^{m(1-t)-j-\frac{k-1}{j}} \leq \frac{2^{k+1+2t}\rho_l^{Mt}}{2^t-2} \sum_{j=2}^M 2^{-jt-\frac{k-1}{j}}.
\end{equation}
Further, we let $M =\Big\lceil 2\sqrt{(k-1)/t}\Big\rceil$.  To have $M \geq 3$, we must restrict $t$ to ${t \leq k-1}$. Also, for $k\geq 9$, we have $M=\Big \lceil 2\sqrt{(k-1)/t} \Big\rceil \leq \Big\lceil 2\sqrt{(k-1)/2} \Big\rceil \leq 2\sqrt{k-1}-1.$
From (\ref{alpha-split-2}), using  (\ref{part1alpha}) and (\ref{part2alpha}) and Lemma \ref{prop:eq:3}, we get
\begin{align}\label{sum-new-est}
    \sum{'}_{n\in M_{k,l}}\overline{\alpha}_D(n)^t &\leq 2^{k-1-t(M-1)t} \rho_l^{(M-2)t} + \frac{2^{k+1+2t-2\sqrt{t(k-1)}}}{2^t-2} \rho_l^{Mt}(M-1) \notag\\
    &\leq 2^{k-1+t-2\sqrt{t(k-1)}} \rho_l^{2\sqrt{t(k-1)}-t} \notag \\
    & \hspace{3mm}+\frac{2^{k+2+2t-2\sqrt{t(k-1)}}}{2^t-2}\sqrt{\frac{k}{t}}\rho_l^{2\sqrt{t(k-1)}+t}  \\
    &=2^{k-1+t-2\sqrt{t(k-1)}}\rho_l^{2\sqrt{t(k-1)}+t}\Bigg(\rho_l^{-2t}+\sqrt{\frac{k}{t}}\frac{2^{3+t}}{2^t-2} \Bigg) \notag \\
    &<2^{k-1+t-2\sqrt{t(k-1)}}\rho_l^{2\sqrt{t(k-1)t}+t}\Bigg(1+\sqrt{\frac{k}{t}}\frac{2^{3+t}}{2^t-2} \Bigg).
\end{align}
As $\frac{2^t}{\sqrt{t}(2^t-2)}$ is monotonically decreasing in $t\geq 2$, we have for $t\geq 2$
$$
\frac{2^{3+t}}{2^t-2}\sqrt{\frac{k}{t}}< \frac{2^5}{2}\sqrt{\frac{k}{2}}= 4^{1.75}\sqrt{k}.
$$
For $k\geq 1$ we have $1+4^{1.75}\sqrt{k}< 4^{1.812}\sqrt{k}.$
For $t \leq (k-1)/9$ we get by Lemma \ref{prop:eq:3} that
$$
2^{\sqrt{t/(k-1)}} \leq 2^{\sqrt{1/9}}=1.25992<1.26.
$$
Thus, we get from (\ref{sum-new-est})
\begin{align*}
    \sum{'}_{n\in M_{k,l}}\overline{\alpha}_D(n)^t \leq 
    2^{k+t}\rho_l^{2\sqrt{kt}+t}4^{1.312-\sqrt{tk}}(1.26)\sqrt{k}.
\end{align*}
Using this in (\ref{probLuc}) with Proposition \ref{k-bit-prime-approx} we get the desired result.
\end{proof}

Tables \ref{table-qkt}, \ref{table-qklt} and \ref{table-pkt} compare the bounds for $q_{k,t}$ using Theorem \ref{k,t theorem}, $q_{k,l,t}$ using Theorem \ref{q-k,t new} and $p_{k,t}$ using Theorem \ref{results-dametal} (\ref{thrm:results2}), where $l$ was chosen with respect to $k$ as discussed in Remark \ref{remark1}.

\begin{table}[H]
\centering
 \begin{tabular}{ l|  l l l l l l}
 $k\backslash t$  & 2 & 4 & 8 & 16 & 32 & 64 \\ \hline
	100  & 8 & 13 & 18 &  & & \\ 
	200  & 17 & 28 & 38 & 44 &  &  \\ 
	400  & 32 & 49 & 68 & 87 & 96 &  \\ 
	512  & 39 & 59 & 82 & 107 & 124 &  \\
	1024 & 64 & 94 & 132 & 178 & 223 & 252 \\ 
	2048 & 99 & 145 & 205 & 280 & 367 & 454 \\
	4096 & 151 & 218 & 308 & 426 & 574 & 745 \\ 
\end{tabular}
\caption{ \label{table-qkt} Lower bounds for $-\log_2(q_{k,t})$ using Theorem \ref{k,t theorem}.}
\end{table}

\begin{table}[H]
\centering
\begin{tabular}{  l |  l l l l l l}
	$k\backslash t$  & 2 & 4 & 8 & 16 & 32 & 64 \\ \hline
	100  & 12 & 22 & 34 &  &  &  \\ 
	200  & 22 & 37 & 56 & 81 &  &  \\ 
	400  & 37 & 59 & 88 & 126 & 176 &  \\ 
	512  & 44 & 69 & 102 & 147 & 205 &  \\ 
	1024  & 69 & 105 & 154 & 221 & 310 & 428 \\ 
	2048  & 105 & 156 & 227 & 325 & 459 & 639 \\ 
	4096  & 157 & 230 & 332 & 474 & 670 & 938 \\ 
\end{tabular}
\caption{\label{table-qklt} Lower bounds for $-\log_2(q_{k,l,t})$ using Theorem \ref{q-k,t new} with $l$ chosen with respect to $k$ defined as in Remark \ref{remark1}.}
\end{table}

\begin{table}[H]
\centering
\begin{tabular}{ l |  l l l l l l }
       	$k\backslash t $ & 2 & 4 & 8 & 16 & 32 & 64 \\ \hline
	100 & 12 & 23 & 36 &  &  &  \\ 
	200 & 23 & 38 & 58 & 83 &  &  \\ 
	400 & 38 & 60 & 89 & 129 & 179 &  \\ 
	512 & 45 & 70 & 104 & 149 & 209 &  \\ 
	1024 & 70 & 106 & 155 & 223 & 313 & 432 \\ 
	2048 & 106 & 157 & 229 & 327 & 462 & 642 \\ 
	4096 & 157 & 231 & 333 & 476 & 672 & 941 \\ 
\end{tabular}
\caption{\label{table-pkt} Lower bounds for $-\log_2(p_{k,t})$ using Theorem \ref{results-dametal} (\ref{thrm:results2}).}
\end{table}

\section{The worst-case numbers}\label{worstcase}

The numbers with the largest $\alpha_D(n)$ contribute most to the probability estimate in our analysis. The sets $C_{1,D}$ and $C_{2,D}$ are empty, as Theorem \ref{rabin-monier for Lucas} states that $\alpha_D(n)\leq 1/4$ for all $n\in \mathbb{N}$. By treating the sets $C_{3,D}, C_{4,D}$ and $C_{5,D}$ separately, we aim to achieve a better estimate for $q_{k,l,t}$ for large $t$. However, as we will see, we encounter the challenge that Lucas-Carmichael numbers belong to this set.  Unfortunately, establishing bounds for these numbers remain an open question in number theory, hindering further process. Once bounds are found, the derivation becomes straightforward.

In this section, we always assume that $\epsilon(n)=-1.$

\subsection{Classifying \texorpdfstring{$C_{3,D}$ }{C3D}}\label{classifyc3}
First, we classify the members of $C_{3,D}$. In this subsection, unless specified otherwise, let $n$ always represent an integer relatively prime to $2D$ with prime decomposition $n=p_1^{r_1}\cdot\ldots \cdot p_s^{r_s}$, and write $n-\epsilon(n)=2^\kappa q$ and $p_i-\epsilon(p_i)=2^{k_i}q_i$, where $q,q_i$ odd, and the prime factors $p_i$ are ordered such that $k_1 \leq \ldots \leq k_s$.

We will later make use of the following lemmas in our proofs.
\begin{lemma}[Arnault \cite{Rabin-Mon-Lucas}]\label{ineqs} 
Let $n$ be as described above. Then
    \begin{equation}
    \frac{SL(D,n)}{\varphi_D(n)} \leq \begin{cases} \frac{1}{2^{s-1}}\prod_{i=1}^s \frac{\gcd(q,q_i)}{q_i}, \\
    \frac{1}{2^{s-1}}\prod_{i=1}^s \frac{1}{p_i^{r_i-1}}, \\
    \frac{1}{2^{s-1+\delta_2+\ldots+\delta_s}}, \text{ where } \delta_i=k_i-k_1.
    \end{cases}
\end{equation}
\end{lemma}

\begin{lemma}\label{ineqs2}
Let $n$ be as described above. Then
$$\frac{SL(D,n))}{\varphi_D(n)}\leq 2^{1-s+\sum_{i=1}^s (k_1-k_i)} \prod_{i=1}^s \frac{\gcd(q,q_i)}{q_i}.$$
\end{lemma}
\begin{proof}
From inequality (\ref{SLbound}) and Lemma \ref{sum_k_1} we get that $$ SL(D,n) \leq  2^{1+(k_1-1)s} \prod_{i=1}^s \gcd(q,q_i).$$  We also have that $$\varphi_D(n) = \prod_{i=1}^s p_i^{r_i-1}(p_i-\epsilon(p_i)) \geq \prod_{i=1}^s (p_i-\epsilon(p_i)) = \prod_{i=0}^s 2^{k_i}q_i.$$
By combining these expressions and seeing that $2^{1+(k_1-1)s} \prod_{i=1}^s 2^{-k_i}= 2^{1-s+\sum_{i=1}^s(k_1-k_i)}$, we get the desired result.
\end{proof}

\begin{lemma}\label{alphaforC_3}
Let $n$ be as described above. Then
\begin{equation*}
    \frac{SL(D,n)}{\varphi_D(n)}=\frac{1}{2^{k_1+k_2+\dots +k_s}} \prod_{i=1}^s\frac{1}{p_i^{r_i-1}} \Bigg( \prod_{i=1}^s \frac{\gcd(q,q_i)-1}{q_i}+ \frac{2^{sk_1}-1}{2^{s}-1}\prod_{i=1}^s \frac{\gcd(q,q_i)}{q_i} \Bigg).
\end{equation*}
\end{lemma}

\begin{proof}
We have
$$
\varphi_D(n)= \prod_{i=1}^s \varphi_D(p_i^{r_i}) = \prod_{i=1}^s p_i^{r_i-1}(2^{k_i}q_i)=2^{k_1+k_2+\dots +k_s} \prod_{i=1}^s q_i \prod_{i=1}^s p_i^{r_i-1}.
$$
Together with
\begin{align*}
SL(D,n) &=\Big( \prod_{i=1}^s \gcd(q,q_i)-1\Big) + \sum_{j=0}^{k_1-1} 2^{js} \prod_{i=1}^s \gcd(q,q_i) \\
&=\Big( \prod_{i=1}^s \gcd(q,q_i)-1\Big) + \frac{2^{sk_1}-1}{2^s-1}\prod_{i=1}^s \gcd(q,q_i),
\end{align*}
we get the desired result.
\end{proof}

\begin{lemma}\label{order-delta-C_3}
Let $n=p_1p_2$ and $\delta_2=k_2-k_1.$ Then
\begin{align*}
    2^kq =2^{2k_1+\delta_2}q_1q_2\pm 2^{k_1}(q_1\pm 2^{\delta_2}q_2).
\end{align*}
\end{lemma}
\begin{proof}
\begin{align*}
    2^kq=&p_1p_2-\epsilon(p_1p_2) = (2^{k_1}q_1+\epsilon(p_1))(2^{k_1+\delta_2}q_2+\epsilon(p_2))-\epsilon(p_1p_2) \\
    =& 2^{2k_1+\delta_2}q_1q_2+2^{k_1}q_1\epsilon(p_2)+2^{k_1+\delta_2}q_2\epsilon(p_1)+\epsilon(p_1)\epsilon(p_2)-\epsilon(p_1p_2) \\
    =& 2^{2k_1+\delta_2}q_1q_2+2^{k_1}(q_1\epsilon(p_2)+2^{\delta_2}q_2\epsilon(p_1)) = 2^{2k_1+\delta_2}q_1q_2\pm 2^{k_1}(q_1\pm 2^{\delta_2}q_2).
\end{align*}
\end{proof}

\begin{lemma}[Arnault \cite{Rabin-Mon-Lucas}]\label{largerthan1/3} Let $n=(2^{k_1}q_1-1)(2^{k_1}q_1+1)$. Then for all $q_1, k_1 \in \mathbb{N}$ with $q_1\neq 1 $ odd we have $\frac{SL(D,n)}{\varphi_D(n)} > \frac{1}{3}.$
For $q_1=1$, we have $\frac{SL(D,n)}{\varphi_D(n)}=\frac{1}{3}-\frac{1}{3\cdot 4^{k_1}}.$
\end{lemma}

Now have now all the ingredients to prove the main theorem of this section. For integers $m, n, \beta$, we mean by $m^\beta \mid \mid n$ that $m^\beta \mid n$ and $m^{\beta+1} \nmid n$.
\begin{theorem}\label{comprise_C_3}
Let $n=p_1^{r_1}\ldots p_s^{r_s}$ be the prime decomposition of an integer $n$ relatively prime to $2D$. Let $n-\epsilon(n)=2^\kappa q$ and $p_i-\epsilon(p_i)=2^{k_i}q_i$, with $q,q_i$ odd, ordering the $p_i$'s such that $k_1 \leq \dots \leq k_s$.
$C_{3,D}$ consists of the following numbers:
\begin{enumerate}
    \item $n=9, 25, 49$.
    \item $n=p_1p_2=
    \begin{cases}
     (2^{k_1}q_1-1)(2^{k_1}q_1+1), \\
    (2^{k_1}q_1+ \epsilon(p_1))(3\cdot 2^{k_1}q_1+\epsilon(p_2)),\\
    (2^{k_1}q_1+ \epsilon(p_1))(2\cdot 2^{k_1}q_1+\epsilon(p_2)) \text{ with } (q_1,k_1) \neq (1,1),
    \end{cases}
$\\ with $k_1\in\N$, $q_1$ odd and each factor is prime.
\item $n=p_1p_2p_3$ is a product of three distinct prime factors, $p_i- \epsilon(p_i) \mid n - \epsilon(n)$ and there is some integer $k_1$ such that $2^{k_1} \mid \mid p_i -\epsilon(p_i)$  for all $i \in \{1,2,3\}.$
\end{enumerate}
\end{theorem}

\begin{proof}\leavevmode

\begin{enumerate}
    \item Let $s=1$, hence $n=p_1^{r_1}$, where $r_1 \geq 2$. By the second inequality of Lemma \ref{ineqs}, we know that $\alpha_D(n) \leq \frac{1}{p_i^{r_i-1}}$. Thus, if $r_1 \geq 3$, then $\alpha_D(n) \leq \frac{1}{9}$ and $n\not \in C_3$. If $r_1=2$, then $\alpha_D(n) \leq \frac{1}{11}$ for $p_i \geq 11.$ Hence, the only possibilities are  $n=9, 25, 49$. 
    Let us check if such an $n\in C_{3,D}$.
    \newline
    Let $n=9$. If $\epsilon(3)=1,$ we have by Lemma \ref{alphaforC_3} that $\alpha_D(9)=\frac{1}{6}$. If $\epsilon(3)=-1$ however, we get by Lemma \ref{alphaforC_3} that $\alpha_D(9)=\frac{1}{4}$. In both cases $9 \in C_{3,D}$.
    \newline
    Let $n=25$. If $\epsilon(5)=1$, we get by Lemma \ref{alphaforC_3} that $\alpha_D(25)=\frac{3}{20}$. If $\epsilon(5)=-1$, we get by Lemma \ref{alphaforC_3} that $\alpha_D(25)=\frac{5}{30}$. In both cases $25\in C_{3,D}.$
    \newline
    Let $n=49$. If $\epsilon(7)=1$, we get by Lemma \ref{alphaforC_3} that $\alpha_D(49)=\frac{5}{42}<\frac{1}{8},$ so such a decomposition of 49 would not be in $C_{3,D}$. If $\epsilon(7)=-1$ however, we get by Lemma \ref{alphaforC_3} that $\alpha_D(49)=\frac{1}{8}$, so in this case $49 \in C_{3,D}.$

    \item Now let $s=2$, hence $n=p_1^{r_1}p_2^{r_2}$. If $p_1=3$, then $r_1 \leq 2$ and $r_2\leq 1$, otherwise by the second inequality of Lemma \ref{ineqs} $\alpha_D(n) \leq \frac{1}{18}$. If $p_1, p_2 \geq 5$, it follows from the second inequality of Lemma \ref{ineqs} that  $r_i =1$, because otherwise  ${\alpha_D(n) \leq \frac{1}{2}\cdot \frac{1}{5}=\frac{1}{10}}$.
    Thus, either $n=p_1p_2$ with $p_1,p_2 > 3$ or $n=3^2p_2$. We shall first treat the case $n=p_1p_2$ with $p_1,p_2 > 3$.

    Now let $n=p_1p_2$ with $p_1 -\epsilon(p_1) =2^{k_1}q_1$ and $p_2 - \epsilon(p_2)=2^{k_2}q_2$.  If $k_2 \geq k_1 +2$, we have that $\alpha_D(n) \leq \frac{1}{8}$ by the third inequality of Lemma \ref{ineqs}. Hence, either $k_2=k_1$ or $k_2=k_1+1.$ 
    
    By the first inequality of Lemma \ref{ineqs} either both $\frac{\gcd(q,q_1)}{q_1}=\frac{\gcd(q,q_2)}{q_2}=1$ or $\frac{\gcd(q,q_i)}{q_i} = \frac{1}{3}$ for exactly one $i$ and $\frac{\gcd(q,q_j)}{q_j} = 1$ for the other $j \neq i$, as otherwise $\alpha_D(n) \leq \frac{1}{18}$.
    
    If $k_2=k_1+1$, it must hold that $\frac{\gcd(q,q_1)}{q_1}=\frac{\gcd(q,q_2)}{q_2}=1$, otherwise by Lemma \ref{ineqs2} we have $\alpha_D(n) \leq \frac{1}{12}$.
    
    Thus, we are left to check the following three cases: The first one is $k_1=k_2$ with $\frac{\gcd(q,q_1)}{q_1}=\frac{\gcd(q,q_2)}{q_2}=1$, the second one is 
 $k_1=k_2$ with $\frac{\gcd(q,q_1)}{q_1}=1$ and $\frac{\gcd(q,q_2)}{q_2}=\frac{1}{3}$, and the third one is $k_2=k_1+1$ with $\frac{\gcd(q,q_1)}{q_1}=\frac{\gcd(q,q_2)}{q_2}=1$.

        Let us look at the case where $k_1=k_2$ with $\frac{\gcd(q,q_1)}{q_1}=\frac{\gcd(q,q_2)}{q_2}=1$. This is equivalent to $q_1, q_2 \mid q$. By Lemma \ref{order-delta-C_3} both $q_1, q_2$ divide 
${2^\kappa q = 2^{2k_1}q_1q_2 \pm 2^{k_1}(q_1\pm q_2).}$
This is only possible if $q_1=q_2$, and thus $p_1 -\epsilon(p_1)=2^{k_1}q_1$ and $p_2-\epsilon(p_2)=2^{k_1}q_1$. In order for $p_1$ and $p_2$ to be distinct primes, we must have that $\epsilon(p_1)\neq \epsilon(p_2)$. Without loss of generality we let $\epsilon(p_1)=1$ and $\epsilon(p_2)=-1$. Therefore, $$n=(2^{k_1}q_1-1)(2^{k_1}q_1+1).$$ Let us check if such an $n$  is indeed in $C_{3,D}$. By Lemma \ref{largerthan1/3} we know that $\frac{SL(D,n)}{\varphi_D(n)} > \frac{1}{3}$ for all odd $q_1 \neq 1$, and for $q_1=1$ we have  $\frac{SL(D,n)}{\varphi_D(n)}=\frac{1}{3}-\frac{1}{3\cdot 4^{k_1}}.$ Since this is monotonically increasing in $k_1$, we have $\alpha_D(n) = \frac{1}{3}-\frac{1}{3\cdot 4^{k_1}} \geq \frac{1}{4}$. Thus, $n \in C_{3,D}$. 

 Now let us look at the case where $k_1=k_2$ with $\frac{\gcd(q,q_1)}{q_1}=1$ and $\frac{\gcd(q,q_2)}{q_2}=\frac{1}{3}.$ Thus,  $q_1$ and $\frac{1}{3}q_2$ both divide $q$, and by Lemma \ref{order-delta-C_3} also $ {2^\kappa q= 2^{2k_1}q_1q_2 \pm 2^{k_1}(q_1\pm q_2).}$
Hence, $q_1 \mid q_2$ and $\frac{1}{3}q_2 \mid q_1$. This implies that there exists an $a\in\N$ such that $q_1\cdot a=q_2$, and a $b\in \N$ such that $\frac{1}{3}q_2b=q_1$. Solving the two equations yields $a=3$ and $b=1$, thus $q_2=3q_2$. Therefore, $p_1-\epsilon(p_1)= 2^{k_1}q_1$ and $p_2-\epsilon(p_2)=2^{k_1}3q_1$. Thus,
\begin{equation*}
    n=(2^{k_1}q_1+\epsilon(p_1))(2^{k_1}3q_1+\epsilon(p_2)).
\end{equation*}
Let us check if such an $n$ is indeed in $C_{3,D}$. By Lemma \ref{alphaforC_3} we have ${\alpha_D(n)= \frac{1}{4^{k_1}} \Big( \Big(\frac{q_1-1}{3q_1}\Big)^2+\frac{4^{k_1}-1}{9} \Big).}$ If $q_1=1$, we have $\alpha_D(n)=\frac{4^{k_1}-1}{9 \cdot 4^{k_1}}< \frac{1}{8}$, so $n \not \in C_3.$ If $q_1 \neq 1$, we have ${\alpha_D(n)=\frac{1}{4^{k_1}} \Big( \frac{1}{3}\Big(\frac{q_1-1}{q_1}\Big)^2+\frac{4^{k_1}-1}{9} \Big) \geq \frac{1}{3\cdot4^{k_1}}\cdot\frac{ 4^{k_1+1}-1}{12}>\frac{1}{8}}$, where we used the fact that both $\frac{q_1-1}{q_1}$ and $\frac{4^{k_1+1}-1}{4^{k_1}}$ are monotonically increasing functions in $q_1$ and $k_1$ respectively. Thus, $n\in C_3.$

Now let us look at the case $k_2=k_1+1$ with $\frac{\gcd(q,q_1)}{q_1}=\frac{\gcd(q,q_2)}{q_2}=1$.  By Lemma \ref{order-delta-C_3} both $q_1, q_2$ divide $ 2^kq = 2^{2k_1+1}q_1q_2\pm 2^{k_1}(q_1\pm 2q_2).$
Thus, $q_1 \mid 2q_2$ and $q_2 \mid q_1$. Since $q_1$ is odd, we must have that $q_1 \mid q_2$, which is only possible when $q_1= q_2$. Hence, $p_1 - \epsilon(p_1) = 2^{k_1}q_1$ and $p_2 - \epsilon(p_2) = 2^{k_1+1}q_1= 2(2^{k_1}q_1)=2(p_1-\epsilon(p_1)).$ Therefore,
\begin{equation*}
    n=p_1p_2=(2^{k_1}q_1+\epsilon(p_1))(2\cdot 2^{k_1}q_1+\epsilon(p_2)).
\end{equation*}
Let us check if such an $n$ is in $C_3.$ By Lemma \ref{alphaforC_3}, we have that  ${\alpha_D(n)=\Big(\frac{q_1-1}{q_1}\Big)^2\frac{1}{2\cdot 4^{k_1}}+ \frac{4^{k_1}-1}{6\cdot 4^{k_1}}.}$ If $q_1=1$, we obtain $\alpha_D(n)=\frac{4^{k_1}-1}{6\cdot 4^{k_1}}$. This is only greater than $\frac{1}{8}$ for $k_1>1$. For $k_1=1$, we obtain $\alpha_D(n)=\frac{1}{8}$, the only possibility is  $n=(2+\epsilon(p_1))(4+\epsilon(p_2))=3\cdot 5$. If $q_1\neq 1$, using the fact that $(q_1-1)/q_1$ is monotonically increasing in $q_1$, we obtain 
\begin{align*}
    \alpha_D(n)=\Big(\frac{q_1-1}{q_1}\Big)^2 \frac{1}{2\cdot 4^{k_1}}+ \frac{4^{k_1}-1}{6\cdot 4^{k_1}} \geq  \frac{4}{18\cdot 4^{k_1}} + \frac{4^{k_1}-1}{6\cdot4^{k_1}} = \frac{4}{18\cdot4^{k_1}} + \frac{1}{8}>\frac{1}{8}.
    \end{align*}

Now let us treat the case $n=3^2p_2$. Since $3-\epsilon(3)=2^{k_1}q_1$, but $\epsilon(3)= \pm 1$, we have that $3-\epsilon(3)\in \{2,4\}$, which implies that $q_1=1$ and $k_1 \in \{ 1,2 \}.$\newline
By the third inequality of Lemma \ref{ineqs} we have for $k_2 \geq k_1 +2$ that $n\not \in C_3,$ thus either $k_1=k_2$ or $k_2=k_1+1$.
Now let $k_2=k_1$. Again it must hold that either $\frac{\gcd(q,q_1)}{q_1}=\frac{\gcd(q,q_2)}{q_2}=1$, or $\frac{\gcd(q,q_1)}{q_1}=1$ and $\frac{\gcd(q,q_2)}{q_2}=3$, since $q_1=1$.
We have
\begin{align}\label{eps-3^2p}
2^\kappa q =& n-\epsilon(n) =3^2p_2-\epsilon(3^2p_2) =(2^{k_1}+\epsilon(3))^2(2^{k_1+\delta_2}q_2+ \epsilon(p_2))-\epsilon(p_2) \notag \\
=&(2^{2k_1}+2^{k_1+1}\epsilon(3)+1)(2^{k_1+\delta_2}q_2+\epsilon(p_2))-\epsilon(p_2)\notag  \\ 
=&q_2(2^{3k_1+\delta_2}+2^{k_1+\delta_2}+\epsilon(3)2^{2k_1+1+\delta_2}) +\epsilon(p_2)(2^{2k_1}+\epsilon(3)2^{k_1+1}).
\end{align}
Now let us look at the case where $\frac{\gcd(q,q_2)}{q_2}=1$, meaning $q_2 \mid q.$  Inequality (\ref{eps-3^2p}) implies that $q_2 \mid 2^{2k_1} + \epsilon(3) 2^{k_1+1}$. If $k_1=1$, we must have that $\epsilon(3)=1$, otherwise $2^{k_1}+\epsilon(3) \neq 3$. Hence, $q_2 \mid  8$. If $k_1=2$, we must have that $\epsilon(3)=-1$, otherwise $2^{k_1}+\epsilon(3) \neq 3$. Hence, $q_2 \mid 8$. Since $q_2$ must be odd, the only possibility is $q_2=1$. 
This analysis holds for both $k_2=k_1$ and $k_2=k_1+1$. Therefore, we get
\begin{equation*}
p_2= 2^{k_2}q_2 + \epsilon(p_2)=\\
    \begin{cases}
    2^{k_1}q_2 \pm 1 = 1, 3,  &\text{ for } k_2=k_1=1, q_2=1, \\
    2^{k_1}q_2 \pm 1 = 3, 5,  &\text{ for } k_2=k_1=2, q_2=1,\\
    2^{k_1+1}q_2 \pm 1 = 3, 5 &\text { for } k_2=k_1+1, k_1=1, q_2=1, \\
     2^{k_1+1}q_2 \pm 1 =  7, 9 &\text { for } k_2=k_1+1, k_1=2, q_2=1.
    \end{cases}
\end{equation*}
Since $p_2$ is a prime different from 3, we discard all other cases and are left with $p_2 \in \{5, 7\}.$ \newline
Now let us look at the case where $\frac{\gcd(q,q_2)}{q_2}=\frac{1}{3}$, meaning ${\frac{1}{3}q_2 \mid q.}$ Here it must hold that $k_1=k_2$. By the same reasoning as above we have ${\frac{1}{3}q_2 \mid 2^{2k_1} + \epsilon(3)2^{k_1+1}}$, which implies $q_2 \mid 3(2^{2k_1} + \epsilon(3) 2^{k_1+1})$. For $k_1=1$ we have $\epsilon(3)=1$ and hence $q_2 \mid 24$. If $k_1=2$, it holds that $\epsilon(3)=-1$ and hence $q_2 \mid 24$. Again since $q_2$ must be odd, the only possibility now is $q_2=3$.
Thus, we get
\begin{equation*}
 p_2 = 2^{k_2}q_2 +\epsilon(p_2) =
    \begin{cases}
    2^{k_1}q_2 \pm 1 =5, 7,  &\text{ for } k_1=1, q_2=3, \\\
    2^{k_1}q_2 \pm 1 = 11, 13,  &\text{ for } k_1=2, q_2=3.
    \end{cases}
\end{equation*}
Again we discard the cases where $p_2$ is not a prime or divisible by 3 and are left with $p_2= 5, 7, 11, 13$.\newline
We see that for $n=3^2p_2$ with $p_2\geq 5$ prime and $n\in C_{3,D}$ the only possibilities are $n=45, 63, 99, 117.$ Now let us check if such an $n \in C_{3,D}.$\newline
Let $n=45$. By the arguments above, there are only three possible decompositions that would make $45\in C_{3,D}$. The first being $\epsilon(5)=1$, $\epsilon(3)=-1$ with $k_1=k_2=2$ and $q_1=q_2=1$, $q=11$. By Lemma \ref{alphaforC_3} this yields $\alpha_D(n)=\frac{5}{48}<\frac{1}{8}$.
The second decomposition is $\epsilon(5)=\epsilon(3)=1$ with $k_1=1, k_2=2$ and $q_1=q_2=1$, $q=11$. Again by Lemma \ref{alphaforC_3} we get $\alpha_D(n)=\frac{1}{24}$.
The third decomposition is $\epsilon(5)=-1$, $\epsilon(3)=1$ with $k_1=k_2=1$, and $q_1=1$, $q_2=3$, $q=23.$ This gives us $\alpha_D(n)=\frac{1}{36}$. In any case $45 \notin C_{3,D}.$\newline
Let $n=63$. By the arguments above, there are only two possible decompositions that would make $63\in C_{3,D}$. The first one being $\epsilon(7)=-1, \epsilon(3)=1$, with $k_1=2, k_2=3$ and $q_1=q_2=1$, $q=1$. By Lemma \ref{alphaforC_3} this yields $\alpha_D(63)=\frac{5}{96}<\frac{1}{8}$.
The second decomposition is $\epsilon(5)=\epsilon(3)=1$ with $k_1=k_2=1$ and $q_1=1, q_2=3$, $q=31$. Again by Lemma \ref{alphaforC_3} we get $\alpha_D(63)=\frac{1}{36}$. In any case $63\notin C_{3,D}.$\newline
Let $n=99$ or $n=117$. By the arguments above, the values for $s, k_1, k_2, q_1, q_2$ and $\gcd(q,q_i)$ for ${i=1,2}$ that would make $n \in C_{3,D}$ are the same. We use Lemma \ref{alphaforC_3} to calculate $\alpha_D(n)$ and get that $\alpha_D(99)=\alpha_D(117)=\frac{5}{144},$ so both ${99, 117 \notin C_{3,D}.}$

\item  Now let $s=3$ with $n=p_1^{r_1}p_2^{r_2}p_3^{r_3}$. By the second inequality of Lemma \ref{ineqs} it must hold that $r_i=1$ for all ${i =1,2,3}$, otherwise $\alpha_D(n) \leq \frac{1}{12}.$ Therefore, $n=p_1p_2p_3$ with $p_i \neq p_j$ for every $i\neq j$. By the first inequality of Lemma \ref{ineqs}, we have that $\frac{\gcd(q,q_i)}{q_i}=1$ for all $i =1,2,3$, otherwise $\alpha_D(n)\leq \frac{1}{12}$. Thus, $q_i \mid q$ for every $i =1,2,3$. By the third inequality of (\ref{ineqs}), we must have that $k_1=k_2=k_3$, as else $\alpha_D(n) \leq \frac{1}{8}.$\newline
Therefore, we have $k_1=k_2=k_3$ with $q_i \mid q$ for all $i \in \{1,2,3\}$, thus also $q_i \mid 2^\kappa q.$ Since $r_i=1$ is odd for all $i$ and also the number of $k_i=\kappa$ is odd, we have that $2^{k_i}q_i \mid 2^\kappa q$, which is the same as $p_i -\epsilon(p_i) \mid n-\epsilon(n)$.

Let us check if such an $n$ is indeed in $C_3.$
Using Lemma \ref{alphaforC_3} and the fact that $k_1=k_2=k_3$, $q_i \mid q$ and $r_i=1$ for $i=1,2,3$, we get
\begin{align*}
\alpha_D(n) = \frac{1}{2^{3k_1}}\Bigg( \prod_{i=1}^3 \frac{q_i-1}{q_i} + \frac{2^{3k_1}-1}{7} \Bigg) = \frac{1}{2^{3k_1}}\prod_{i=1}^3 \frac{q_i-1}{q_i}+ \frac{1}{7}\cdot \frac{2^{3k_1}-1}{2^{3k_1}}.
\end{align*}
Since $\frac{2^{3k_1}-1}{{3k_1}}$ is monotonically increasing in $k_1$, we get $\frac{2^{3k_1}-1}{{2^{3k_1}}} \geq \frac{2^{3}-1}{2^{3}}=\frac{7}{8}$. Thus
\begin{align*}
\alpha_D(n) = \frac{1}{2^{3k_1}}\prod_{i=1}^3 \frac{q_i-1}{q_i}+ \frac{1}{7}\cdot \frac{2^{3k_1}-1}{2^{3k_1}} \geq \frac{1}{2^{3k_1}}\prod_{i=1}^3 \frac{q_i-1}{q_i}+ \frac{1}{8} > \frac{1}{8}.
\end{align*}
With this we indeed have that such an $n\in C_3.$

\item Now let $s\geq 4$. By the second inequality of Lemma \ref{ineqs} we immediately have that $\alpha_D(n)\leq \frac{1}{8}$, thus $n\not \in C_3$.
 \end{enumerate}
\end{proof}

\subsection{Twin-prime products}

By Theorem \ref{comprise_C_3}, we know that if $n=(2^{k_1}q_1-1)(2^{k_1}q_1+1)$, where both factors are prime, then $n$ belongs to $C_{3,D}$. This corresponds to a subset of the set of products of twin-primes. 

Let $\pi_2(x)= \mid \{ p \leq x : \Omega(p+2)=1 \} \mid$ denote the twin-prime counting function, which counts the number of twin-prime tuples up to $x$. The following theorem provides a bound on the number of twin-primes for $x>e^{42}.$

\begin{theorem}[Riesel, Vaughan \cite{twin-prime-riesel}]\label{riesel-vaughan}
For $x>e^{42}$, we have
\begin{equation*}
    \pi_2(x) < \frac{16 \alpha x}{(7.5+\log(x))\log(x)}, 
\end{equation*}
where $\alpha$ is called the Twin Prime Constant,
\begin{equation*}
    \alpha = \prod_{p>2}\Bigg(1- \frac{1}{(p-2)^2}\Bigg)= \prod_{p>2} \frac{p(p-2)}{(p-1)^2} \approx 0.6602 \dots \end{equation*}
\end{theorem}

Using Theorem \ref{riesel-vaughan}, we bound the number of $k$-bit twin-prime products in the next theorem.

\begin{lemma}\label{k-bit-twin-prime-approx}
For $k\geq 122$ there exists less than $6 \frac{2^{k/2}}{k^2}$ $k$-bit integers that are twin-prime products.
\end{lemma}
\begin{proof}
As $n=p(p+2)$ is a $k$-bit integer, $p$ must be a $k/2$-bit integer. Thus, we only have to consider the number of twin-primes up to $2^{k/2}$. With Theorem \ref{riesel-vaughan} we obtain
\begin{align*}
 \pi_2(2^{k/2}) < \frac{16\alpha 2^{k/2}}{(7.5+\log(2^{k/2}))\log(2^{k/2})} < \frac{16\alpha}{4\log^2(2)} \frac{2^{k/2}}{k^2}< 6\frac{2^{k/2}}{k^2},
\end{align*}
which holds for $2^{k/2} > e^{42}$, so that $k \geq 122$.
\end{proof}

\subsection{Lucas-Carmichael numbers with three prime factors}

Let us analyze the numbers of the third form in Theorem \ref{comprise_C_3}. In fact they are not arbitrary integers, but have already been classified.

\begin{definition}
    Let $D$ be a fixed integer. Let $n$ be an odd composite integer with the property that
    \begin{align}\label{lucascarmichael}
    \begin{split}
        & \textnormal{ for all } P, Q\in \mathbb{N} \textnormal{ with }\gcd(P,Q)=1, P^2-4Q=D \textnormal{ and } \gcd(n,QD)=1,\\
        &\textnormal{ we have } U_{n-\epsilon(n)}(P,Q) \equiv 0 \bmod n.
    \end{split}
    \end{align}
    We call such an $n$ a \textit{Lucas-Carmichael number.}
\end{definition}

A Carmichael number is an odd composite integer $n$ that satisfies ${a^{n-1} \equiv 1 \bmod n}$ for all $a$ such that $\gcd(a,n)=1$. It was shown by Carmichael \cite{carmichael1910note} that if $n$ is a Carmichael number, then $n$ can be expressed as the product of $k\geq 2$ distinct primes $n=p_1, p_2, \ldots, p_k$ and $p_i -1 \mid n-1$ for all $i=1,2,\ldots, k$. Carmichael numbers are a special set of odd composites that pass Fermat's little Theorem, which is the weak version of the Miller-Rabin test, for all suitable values of $a$. In a similar vein, Lucas-Carmichael numbers are the set of odd composites that satisfy (\ref{lucascarmichael}), which is the weaker primality test of the strong Lucas test, for all appropriate pairs $(P,Q)$. Interestingly, if $D=1$ and $n$ satisfies property (\ref{lucascarmichael}), it can be shown that $n$ is a Carmichael number. In that sense Lucas-Carmichael numbers can be seen as a generalization of Carmichael numbers. In 1977, Williams \cite{williams1977numbers} established the following theorem, further reinforcing the connection between Carmichael and Lucas-Carmichael numbers.

\begin{theorem}[Williams \cite{williams1977numbers}]
    Let $D$ be fixed. If $n$ possesses property (\ref{lucascarmichael}), then $n$ is a product of $k$ distinct primes $p_1, p_2, \ldots, p_k$ and
    $$
    p_i - \epsilon(p_i) \mid n - \epsilon(n) \hspace{7mm} \textnormal{ for all } i=1,2,\ldots,k.
    $$
\end{theorem}

Thus, we see that the numbers of the third form in Theorem \ref{comprise_C_3} are exactly the Lucas-Carmichael numbers with three prime factors with the additional property that there exists some $k_1 \in \N$ such that $2^{k_1} \mid \mid p-\epsilon(p_i)$ for all prime factors $p_i$ of $n$. 
Bounding the number of Lucas-Carmichael numbers less than a given integer is an open problem in number-theory, hence we are not able to proceed further.

The question of the existence of an infinite number of Carmichael-Lucas numbers with respect to a fixed $D$ is also an open question. It's worth noting that if $n$ is a Carmichael-Lucas number with respect to either $D = 1$ or $D$ a perfect square, then it is a Carmichael number. Thus, any result in this direction would generalize the result concerning Carmichael numbers in \cite{infiniteCar}, which took 84 years to prove.

\section{Conclusion}\label{sec13}

In this paper, we have successfully established the framework for determining average case error bounds for the strong Lucas test, which was previously unexplored. This is a significant result as it demonstrates the reliability of the strong Lucas test for almost all practical purposes. We have examined an algorithm that randomly chooses $k$-bit integers at random from the uniform distribution, performs $t$ independent iterations of the strong Lucas test on this integer and outputs the first number that passes all $t$ tests. Let $q_{k,t}$ be the probability that this algorithm outputs a composite. The bounds we have derived are $q_{k,1} \leq \log(k)k^24^{2.3-\sqrt{k}}$ for $k\geq 2$ and $q_{k,t}<\log^t(k) \frac{k^{3/2}}{\sqrt{t}}4^{2.12+t-\sqrt{tk}}$ for $k\geq 21$ and $3\leq t\geq (k-1)/9$ or $k\geq 88$ and $t=2$. 
Additionally, we have taken advantage of the computational efficiency of trial division by small primes to rule out candidates before the strong Lucas test and incorporated this into our analysis. By imposing the requirement of checking for divisibility by the first $l$ odd primes before running the strong Lucas test, we have obtained improved bounds. Let $q_{k,l,t}$ be the probability that this updated algorithm returns a composite. Let $\tilde{p}_l$ denote the $l$-th odd prime and let $\rho_l=1+ \frac{1}{\tilde{p}_{l+1}}$.
We have shown that $q_{k,l,1}< k^24^{1.8-\sqrt{k}}  \rho_l^{2\sqrt{k-1}-2}$ for all $l\in \N$ and $k\geq 1$, and $q_{k,l,t}\leq 4^{1.72-\sqrt{tk}}k^{3/2} 2^t \rho_l^{2\sqrt{kt}+t}$ for all $k\geq 21$, $2 \leq t \leq (k-1)/9$, and $l\in \N.$ These bounds are comparable to those of the Miller-Rabin test presented in \cite{DamEtAl}.

Furthermore, we have classified the numbers that contribute most to our probability estimate and identified Lucas-Carmichael numbers with three prime factors are part of this set. Unfortunately, bounding these numbers remains an open question, preventing us from further progress in this regard.

Although we have achieved average case error bounds for the strong Lucas test during the scope of this work, there are still several open questions that look promising for future research. For instance, it would be interesting to extend our average case error estimates to include averaging over both $D$ and $n$. Additionally, investigating error bounds for the incremental search approach in finding primes from a random starting point could be of interest. Moreover, exploring the possibility of obtaining improved estimates for the Miller-Rabin test using the modified algorithm that includes division by small primes is another potential area of study. Finally, obtaining average case error bounds for the Baillie-PSW test, a probabilistic primality test combining the Miller-Rabin test and the strong Lucas test, could be the most exciting future work, given the absence of counterexamples for composites passing this test.

\section*{Acknowledgments}
Many thanks to Mia Filić for all the fruitful discussions.

\end{sloppypar}

\end{document}